\documentclass[10pt, a4]{article}
\usepackage[left=2.cm, right=2.cm, top=2cm, bottom=2cm]{geometry}
\usepackage{amsmath, amsthm}
\usepackage{xspace}
\usepackage{mathtools}
\usepackage{bbm}
\usepackage{bm}
\usepackage{amssymb}
\usepackage{enumitem}
\usepackage{todonotes}
\usetikzlibrary{calc, decorations.pathreplacing, arrows}

\usepackage[noend]{algorithm2e}
\usepackage{hyperref}
\usepackage{authblk}
\usepackage{cleveref}
\usepackage{multirow}
\usepackage{makecell}
\usepackage{xcolor, colortbl}

\newtheorem{theorem}{Theorem}
\newtheorem{remark}[theorem]{Remark}

\newtheorem{lemma}[theorem]{Lemma}
\newtheorem{definition}[theorem]{Definition}
\newtheorem{corollary}[theorem]{Corollary}

\theoremstyle{definition}

\newtheorem*{claim*}{Claim}

\newenvironment{claimproof}[1][\proofname]{
  \pushQED{\qed}%
  \item[ \ensuremath{\vartriangleright}\,
    ]
}{%
  
  \popQED\endtrivlist
  
}

\title{Interval-Constrained Bipartite Matching over Time}

\author[1]{Andreas Abels}
\author[2]{Mariia Anapolska}
\author[2]{Christina Büsing}
\affil[1]{Heinrich-Heine-Universität Düsseldorf, Germany}
\affil[2]{RWTH Aachen University, Germany}

\date{}

\newcommand{\pat}{\ensuremath{j}\xspace}
\newcommand{\ppat}{\pat'} 

\newcommand{\patset}{\mathcal{J}}
\newcommand{\textpat}{job\xspace}

\newcommand{\tp}{\ensuremath{t}\xspace}

\newcommand{\patTW}[1][\pat]{[\earlT[#1], \lateT[#1]]}
\newcommand{\tpmin}{\tp^{-}}
\newcommand{\tpmax}{\tp^{+}}
\newcommand{\earlT}[1][\pat]{r_{#1}}
\newcommand{\lateT}[1][\pat]{d_{#1}}
\newcommand{\req}[1][\pat]{a_{#1}}

\newcommand{\ass}{\alpha}
\newcommand{\assx}{\beta}

\newcommand{\probname}{ICBMT\xspace}
\newcommand{\bmt}{{BMT}\xspace}
\newcommand{\fullprobname}{Interval-Constrained Bipartite Matching over Time}
\newcommand{\FF}{\text{FirstFit}\xspace}
\newcommand{\edf}{EDF\xspace}
\newcommand{\obm}{online bipartite matching\xspace}

\newcommand{\abs}[1]{\left|#1\right|}

\newcommand{\N}{\ensuremath{\mathbb{N}}}
\newcommand{\Q}{\ensuremath{\mathbb{Q}}}

\renewcommand{\O}{\ensuremath{\mathcal{O}}}

\newcommand{\oneto}[1]{\{1,\ldots, #1\}}

\newcommand{\Inst}{\mathcal{I}}

\renewcommand{\geq}{\geqslant}
\renewcommand{\leq}{\leqslant}

\DeclareMathOperator*{\argmin}{arg\,min}
\DeclareMathOperator*{\argmax}{arg\,max}

\newcommand{\blue}[1]{\textcolor{black}{#1}}

\newcommand{\rev}[1]{\textcolor{black}{#1}}

\begin{document}

\maketitle

\begin{abstract}
Interval-constrained online bipartite matching problem frequently occurs in medical appointment scheduling:
Unit-time jobs representing patients arrive online and are assigned to a time slot within their given feasible time interval.
We consider a variant of this problem where reassignments are allowed and extend it by a notion of time that is decoupled from the job arrival events. As jobs appear, the current point in time gradually advances, and once the time of a slot is passed, the job assigned to it is fixed and cannot be reassigned anymore.

We analyze two algorithms for the problem with respect to the resulting matching size and the number of reassignments they make. 
We show that \FF with reassignments according to the shortest augmenting path rule is $\frac{2}{3}$-competitive with respect to the matching cardinality, and that the bound is tight. 
For the number of reassignments performed by the algorithm, we show that it is in $\Omega(n \log n)$ in the worst case, where $n$ is the number of patients or jobs on the online side. 
Interestingly, the competitive ratio remains bounded by $\frac{2}{3}$ if we restrict the algorithm to only consider alternating paths that make up to a constant number $k \geq 1$ of reassignments. 
The latter approach only requires a linear number of reassignments in total. 
This fills the gap between the known optimal algorithm \rev{ that makes no reassignments}, which is $\frac{1}{2}$-competitive, on the one hand, and an earliest-deadline-first strategy (\edf), which we prove to obtain a maximum matching in this over-time framework, but which suffers $\Omega(n^2)$ reassignments in the worst case, on the other hand.

Further, we consider the setting in which the sets of feasible slots per job that are not intervals.
We show that \FF remains $\frac{2}{3}$-competitive in this case, and that this is the best possible deterministic competitive ratio, while \edf loses its optimality.

\end{abstract}
\section{Introduction}
Consider the process of assigning medical appointments of equal duration to patients. 
Patients request appointments one by one and expect an immediate response.
Patient's health condition determines the time interval within which they require an appointment: the urgency of treatment imposes the deadline, and the amount of time needed to prepare the patient for treatment, e.g., additional diagnostics before a surgery, determines the earliest possible appointment day. 
When urgent requests, that is, those with a short feasible time interval, arrive, previously scheduled patients with looser deadlines may be postponed in favor of urgent patients.
As time passes, a growing part of the schedule is in the past and therefore cannot be changed, whereas appointment slots in the future can still be reassigned to different patients. 
The goal is to schedule as many patients as possible and minimize
changes to the schedule.

The described appointment scheduling cannot be modeled precisely as a classical online job scheduling problem, since the decision on whether to reject a job must be made upon job's arrival, not at its latest possible start time.
Instead, we model the process as a variant of \obm with reassignments. In \obm, vertices of one part of the graph, here called “slots”, are known in advance, and vertices of the other part, here called “\textpat{}s”, which represent patients, are revealed one by one together with their incident edges. 
Upon arrival, each \textpat is either immediately assigned to an adjacent slot or irrevocably rejected. It is possible for the algorithm to reallocate, i.e.,~reassign, \textpat{}s to different slots to make place for a new \textpat{}. 

In \obm with reassignments, any reassignment  along the graph edges is feasible. 
In our application, however, if appointment scheduling continues as the appointments take place, reassigning a \textpat to an already passed slot is impossible, and the \textpat assigned to it also cannot be reassigned.
Therefore, we extend this well-studied online problem by a notion of time. 
As time passes, more \textpat{}s arrive and slots are processed, so that once a slot is in the past,
\textpat{}s assigned to it cannot be reassigned anymore.
The task is to maximize the number of assigned \textpat{}s and minimize the number of reassignments.
We call the described problem Bipartite Matching over Time (\bmt).
Furthermore, our \textpat{}s feature interval constraints:
a job can be allocated only to a slot between its given earliest and latest slot.
This yields a special version of the problem, called \emph{\fullprobname} (\probname), which is the main focus of this work.

We analyze two algorithmic approaches to \probname with respect to the number of scheduled jobs and of required reassignments: a greedy \FF algorithm that uses shortest augmenting paths to augment the assignment when there is no feasible slot available, which is analogous to an optimal assignment strategy for \obm, and the earliest-deadline-first strategy, which is motivated by an optimal algorithm for related machine scheduling problems. 
We analyze both approaches in the \probname setting first and then also in the \bmt setting. 

\subsection{Related work}

\probname is \rev{strongly related to} the online bipartite matching problem and the online assignment problem. In these problems, the servers (possibly with capacity greater than one) are known in advance, and clients are revealed online, each with a subset of servers the client can be served by. 
Typical objectives are to maximize the number of assigned clients \cite{Karp1990OnlineBM} or to balance the load of the servers~\cite{BernsteinKPPS17}. 

In the classical maximum \obm, all matching decisions are irrevocable, and the aim is to maximize the number of matched clients. 
No deterministic algorithm can have a competitive ratio greater than $\frac{1}{2}$, which is achieved by a greedy algorithm. Karp, Vazirani, and Vazirani present a randomized algorithm \textsc{ranking}, which is strictly $(1-\frac{1}{e})$-competitive, and show the tightness of the bound \cite{Karp1990OnlineBM}.
The proof has later been corrected, refined, and generalized \cite{GoelM08,BirnbaumM08,DevanurJK13}.

In many applications, it is reasonable to allow reassignments; that is, client rejections are irrevocable, but the assigned server of an already assigned client can be changed. 
In this problem, it is trivial to achieve a maximum matching at the cost of a large number of reassignments.
One widely studied strategy for \obm with reassignments is Shortest Augmenting Paths (SAP), which is also the basis of an optimal algorithm for offline maximum bipartite matching \cite{HopcroftK73}.
Grove et al.~show a lower bound of $\Omega(n \log n)$ reassignments for SAP \cite{Grove1995}. 
They also show a matching upper bound of $\O(n \log n)$ reassignments for the case of all clients having a degree of two. 
This upper bound for SAP is later also proven for arrivals in random order \cite{Chaudhuri2009} and for trees \cite{BosekLSZ22}.
For general graphs, Bernstein, Holm, and Rotenberg show an amortized upper bound of $\O(n \log^2 n)$ reassignments for SAP \cite{Bernstein2019}. Chaudhuri et al.~conjecture that the lower bound is tight for general bipartite graphs \cite{Chaudhuri2009}.

Further exact approaches to \obm with reassignments, which are different from SAP, have also been studied \cite{BosekLSZ14,Chaudhuri2009} and do not improve the number of reassignments.

Another branch of research concerns approximation algorithms for \obm. Bosek et al.~construct $(1-\epsilon)$-approximate matchings applying $\O((1-\epsilon)n)$ reassignments \cite{BosekLSZ14}. Bernstein et al.~approximate the problem for every $p$-norm simultaneously. 
Their algorithm is $8$-competitive with $O(n)$ reassignments~\cite{BernsteinKPPS17}.
Gupta, Kumar, and Stein consider the assignment version of the problem, where each server can serve multiple clients, and present a $2$-competitive algorithm for maximum load minimization with at most $O(n)$ reassignments~\cite{Gupta2013}. 
Recently, Shin et al.~consider the \obm problem under a bounded number $k$ of allowed reassignments per arrival~\cite{Shin2020}. 
They show that a modified SAP-algorithm is $(1-\frac{2}{k+2})$-competitive for this setting and that this bound is tight already on paths.

The influence of recourse has also been studied for more general matching problems, such as minimum-cost matching \cite{MegowN20,GuptaKS20} and maximum matching on general graphs \cite{AngelopoulosDJ20,BernsteinD2020,LiuT22}.

Another closely related problem is unit job scheduling with integer release times and deadlines. The interval constraints directly correspond to release times and deadlines, and the maximum cardinality objective corresponds to throughput maximization. Here, jobs can be buffered until they are executed or their deadline arrives. The offline version of this problem is optimally solved by the Earliest Deadline First strategy, which can be implemented in $O(n)$ time \cite{Frederickson83,SteinerY1993}. 
For the online version, Baruah, Haritsa and Sharma show that any algorithm that is never idle is at least 2-competitive \cite{BaruahHS94}. Chrobak et al.~give a barely random algorithm, an algorithm that only uses a single random bit, that is $\frac{5}{3}$-competitive \cite{ChrobakJST07}. Note that in these scheduling problems, deadlines and release dates are not necessarily integers.
Moreover, in contrast to \probname, decisions on rejections of jobs do not have to be made immediately upon job's arrival, as jobs can remain pending for several time steps.

\subsection{Our contribution}
We introduce a new variant of online bipartite matching: Interval-Constrained Bipartite Matching over Time (\probname). In contrast to related problems that use recourse, the over-time nature of our problem eliminates some of the augmenting paths that may be needed to recover an offline optimal solution to the problem.

We first consider a FirstFit-type algorithm for \probname that uses a shortest augmenting path in case there is no directly available slot, and assigns the earliest free slot otherwise~\rev{(Sec.~\ref{sec:ff})}. 
\FF does not only mimic the common first-come-first-serve practice in appointment scheduling, but is also a natural choice. Due to the over-time property of the problem, algorithms that leave early slots unoccupied risk missing a possible match or have to augment the matching for the current time slot. 
For this algorithm, we give a tight lower and upper bound of $\frac{2}{3}$ on the competitive ratio with respect to the cardinality of the final matching. 
In addition, we give an instance on which the algorithm performs $\Omega(n \log n)$ reassignments. To show this result, we have to specify the algorithm more precisely. For every new assignment, there might be several shortest augmenting paths ending in the same slot, and the choice of one impacts the result. Our algorithm always selects the lexicographically smallest shortest path, where the order is with respect to the sequence of slots in the path. 
This is a natural extension of the \FF strategy, which also favors early slots for the current job. 
We conjecture that there is no deterministic strategy for selecting the shortest paths that yields strictly better results, but a randomization over different strategies might be promising.

Our results imply that, by allowing only a constant number of reassignments for each new assignment, we retain the competitive ratio of $\frac{2}{3}$ for only a linear number of reassignments in total \rev{(Sec.~\ref{sec:k-ff})}. At first glance, this result sounds amazing, but one should not mistake this for the superior algorithm. For every number $k\in\N$ of reassignments allowed per single new assignment, there exists an instance on which this restricted algorithm $k$-\FF only achieves a {$\frac{k+1}{k+2}$-fraction} of the assignment constructed by the unconstrained \FF strategy on the same instance.

These properties of \FF are in contrast to the performance of the earliest-deadline-first strategy: it guarantees a maximum matching; that is, it is $1$-com\-pe\-titive with respect to the cardinality of the final matching~\rev{(Sec.~\ref{sec:edf})}. 
However, this strategy makes $\Omega(n^2)$ reassignments in the worst case, and restricting the number of reassignments per \rev{job arrival} makes the algorithm uncompetitive.

For \bmt, in which the sets of feasible slots for each job are not necessa\-ri\-ly intervals, we show that \FF  and $k$-\FF remain $\frac{2}{3}$-competitive with the same number of reassignments, and that it is the best possible competitive ratio any deterministic algorithm can attain. 
In contrast, \edf becomes only $\frac{1}{2}$-competitive~\Cref{sec:bmt}.

Our main results are summarized in Table~\ref{tab:overview}. 
\rev{All competitive ratios are tight.}
In addition, we analyze the performance of the algorithms in two simpler special cases: uniform interval lengths of jobs and a \rev{common arrival time} for all jobs.


\begin{table}[bt]
\centering
\renewcommand{\arraystretch}{1.1}
\renewcommand\cellgape{\Gape[3pt]}
\begin{tabular}{l|l|c|c|c}
      \multicolumn{2}{l|}{  } & \FF & $k$-\FF & \edf \\[0.8ex]\hline
     \multirow{2}{*}{\makecell[cl]{assigned \textpat{}s \\ (competitive ratio)}} & interval-constrained & {$\frac{2}{3}$} & \makecell[cc]{$\frac{2}{3}$} & \bm{$1$}\\
     \cline{2-5}
      & unconstrained &\bm{$\frac{2}{3}$}&  {\makecell[cc]{\bm{$\frac{2}{3}$}}}&$\frac{1}{2}$\\
      \hline     \multicolumn{2}{l|}{reassignments (both settings)}& \makecell[cc]{$ \Omega(n\log n)$} & $\Theta(n)$& $\Theta(n^2)$
\end{tabular}
\caption{Performance bounds for algorithms presented in this work on instances with $n$ \textpat{}s. Competitive ratios in bold are \rev{the best possible for a deterministic algorithm} in the corresponding setting.
\rev{For \FF, no matching upper bound on worst-case number of reassignments is known.}
}
\label{tab:overview}
\end{table}

\subsection{Problem definition and preliminaries}\label{sec:definitions}

In \emph{\fullprobname}, we are given processing capacity of one slot per time unit~\mbox{$\tp \in \N$}; we identify slots with time units they belong to.
Jobs $\pat\in\{1, \ldots, n\}$ are revealed one by one together with their arrival times $\req \in \N_0$ and their intervals $[\earlT, \lateT]$, $\earlT, \lateT \in \N$, which is the set of their feasible slots.
We assume all numbers to be polynomial in the (unknown) number of jobs.
The arrival times are non-decreasing in the order of \textpat{}s' arrival, i.e., $\req \leq \req[\pat+1]$ for any $\pat< n$.
Several jobs can have the same arrival time, i.e., arrive within the same time unit, but not simultaneously; the order in which they arrive is fixed and, upon arrival of the first such job, the remaining jobs with the same arrival time are not yet visible to the algorithm.
\rev{This assumption models the classical online setting.
The case when all requests arriving at one day are scheduled simultaneously is considered in~\Cref{sec:batching}.}
For all \textpat{}s $\pat \leq n$, the \emph{preparation time} $\earlT - \req$ is at least $1$%
\rev{, as it is a common in-hospital practice to fix a day schedule the day before.
However; all results in this work also hold for nonnegative preparation times.}
The jobs will be denoted as tuples of their attributes $\pat = (\req, \earlT, \lateT)$, and instances as sequences of jobs.

Upon arrival, \textpat $\pat$ must be assigned to some slot $\tp \in [\earlT, \lateT]$.
The already scheduled \textpat{}s can be reassigned
to different slots to free a slot in the \textpat $\pat$'s interval, but they cannot be rejected.
If no reassignment combination allows to assign \textpat \pat, the \textpat is 
irrevocably rejected. 
The assignments in slots $\tp \leq \req$ are fixed and cannot be involved in reassignments for \textpat $\pat$.
This models the fact that the appointment schedule for the days already past, including the current day, cannot be changed. We therefore call the slots up to~$\req$ \emph{fixed}.

The goal is to produce online an assignment
$\ass\colon\oneto{n} \to \N\cup\{\bot \}$ of \textpat{}s to slots, where $\ass(\pat) = \bot$ denotes that the job is rejected.
The assignment should
optimize the following objectives: (1) maximize the number of assigned \textpat{}s, and (2) minimize the number of reassignments.
We consider the two objectives separately.

Any instance of \probname has an \emph{underlying bipartite graph} consisting of the \textpat set on one side and all feasible slots on the other, with edges connecting \textpat{}s to their feasible slots.
    An \emph{offline augmenting path for a solution $\ass \colon \patset\to\N\cup\{\bot\}$} of \probname is 
    an augmenting path with respect to the matching $\ass$ in the underlying bipartite graph. 
\begin{remark}\label{rmk:offlineAugmPath}
An \probname solution has an optimal number of assigned jobs if and only if it admits no offline augmenting paths.
\end{remark}%
In fact, an optimal offline solution for \probname is a maximum matching in the underlying bipartite graph.

The main difference between \probname and the classical \obm with reassignments is the temporal component of \probname, which leads to a partially fixed assignment as time progresses.
However, if all \textpat{}s in an instance arrive at the same time,
then \probname with the objective of minimizing rejections reduces to the online maximum bipartite matching. 
We call this special case Static \probname problem.
\begin{definition}
    \emph{Static \probname} is a special case of the \probname problem, in which all \textpat{}s of an instance have the same arrival time, i.e., for any instance $I = \left((\req, \earlT, \lateT)\right)_{\pat = 1}^{n}$ of Static \probname holds
    $    \req = a \ \text{for all}\ \pat \in \{1,\ldots,n\}$.
\end{definition}
    
Any instance of Static \probname is also an instance of \obm with reassignments: Indeed, since the arrival time of \textpat{}s is constant, no part of the assignment becomes fixed, and all reassignments are allowed at any time.
Hence, a feasible solution for \obm is also a feasible solution for Static \probname on the same set of jobs. 

\newcommand{\tpaux}{\tp^{*}}
\newcommand{\tpx}{\tp_{0}}

\section{Greedy approach: \FF}\label{sec:ff}
One natural and widely used appointment assignment strategy is to greedily assign the earliest possible slot to every patient.
If necessary, a few other patients are reassigned.

The algorithm we call \FF in this work follows this greedy paradigm: Assign \textpat{}s to early slots and minimize the number of reassignments for each job arrival.
It is a special case of the greedy {shortest augmenting path} strategy for bipartite matching that additionally considers the special structure of instances of \probname, namely the total order on the set of slots. 
The algorithm is defined as follows (cf.~Alg.~\ref{alg:FFgen}).

The input consists of a \textpat $\ppat = (\req[\ppat], \earlT[\ppat], \lateT[\ppat])$ and a current slot assignment $\ass\colon\patset\to\N$,
 where {$\patset\subseteq \oneto{n}$} is the set of previously assigned \textpat{}s.
	\FF searches for a shortest sequence of reassignments that frees up one slot in the interval of $\ppat$ and assigns the last \textpat of the sequence to a free slot $\tpaux$. 
	If no such sequence of reassignments exists, \textpat $\ppat$ is rejected.

The sequence of reassignments results from a \emph{shortest augmenting path} $P$ from \textpat $\ppat$ to a free slot $\tpaux$ in a bipartite \textpat-slot graph $G_{\ass, \ppat}$. 
The graph $G_{\ass, \ppat} = (V^T \,\dot\cup\,V^J, E^\ass\,\dot\cup\,E^{\overline{\ass}})$ is a digraph containing a vertex 
for each non-fixed slot and job: 
$V^T = \left\{\tp \in \cup_{\pat\in V^J} \patTW[\pat] \mid \tp > \req[\ppat]\right\}$ and
$V^J = \left\{\pat\in \patset\mid \ass(\pat)> \req[\ppat]\right\}\,\dot\cup\,\{\ppat\}$. 
There are edges~\mbox{$E^{\overline{\ass}}$} from each job to its feasible slots it is not assigned to, 
$E^{\overline{\ass}} = \{(\pat, t) \mid \pat\in V^J,\ t\in V^T,\ t\in \patTW,\ \ass(\pat) \neq t \}$, 
and edges $E^\ass$ from each filled slot to the job assigned to it,
$E^\ass = \{(t, \pat)\mid t \in V^T,\ \pat \in V^J,\ \ass(\pat) = t\}$.
The path $P$ is augmenting with respect to the current assignment $\ass$.
Among multiple possible shortest paths, we choose the one with \emph{earliest target slot}, which motivates the name \FF.
If such a path exists, then the new assignment~$\ass'$ is a symmetric difference between the edges of the old assignment and of the path $P$, denoted $\ass \Delta \{e \mid e\in P\}$.

\RestyleAlgo{ruled}
\SetAlgoSkip{medskip}
\begin{algorithm}[tb]
\caption{\FF}\label{alg:FFgen}
\KwIn{\textpat $\ppat = (\req[\ppat], \earlT[\ppat], \lateT[\ppat])$, current slot assignment $\ass\colon\patset\to\N$}
\KwOut{a new assignment $\ass'\colon\patset\cup\{\ppat\} \to\N\cup\{\bot\}$}
    \vspace*{1mm}
    $F \coloneqq \{\tp \in \N \mid \tp>\req[\pat']\, \land\, \ass^{-1}(\tp) = \emptyset\}$\tcc*[r]{\small{free slots}}
    $\mathcal{P} \coloneqq \{p \mid  p \text{ is a shortest augmenting path from } \ppat \text{ to } F\}$ \;
    $P \coloneqq \argmin\{ \tp \in F \mid \tp \text{ is the endpoint of }p  \in \mathcal{P}\}$ \tcc*[r]{\small{earliest target slot}}
    
    \eIf{$P$ exists}{
    $\ass'\coloneqq \ass \,\Delta\, \left\{(\pat,\tp) \mid (\pat,\tp) \in P \text{ or } (\tp,\pat) \in P\right\}$\;
    }{
    $\ass' \coloneqq \ass \cup (\ppat, \bot)$\;
    }
    \KwRet{$\ass'$}
\end{algorithm}
Note that we do not need to enumerate the set of paths $\mathcal{P}$: the earliest-target shortest augmenting path can be found efficiently via breadth-first search on $G_{\ass, \ppat}$, if the nodes of one level are processed in the chronological order. 
Notice that \FF can also be used for Bipartite Matching over Time without interval constraints, where sets of feasible slots of \textpat{}s are of arbitrary structure (see \Cref{sec:bmt}).
On interval-constrained instances, however, we additionally have the following property.

\begin{remark}\label{rmk:targetslot}
      When assigning job $\ppat$, \FF for \probname always chooses the earliest free slot in the current assignment in or after the interval of \textpat $\ppat$, i.e., slot~\mbox{$\tpaux\coloneqq \min\{\tp \geq \earlT[\ppat] \mid \ass^{-1}(\tp) = \emptyset\}$}, as the target free slot.
\end{remark}
\begin{proof}
We show that no other free slot can be reached from job $\ppat$ by a shorter augmenting path  than a shortest augmenting $\ppat$-$\tpaux$ path.

If slot $\tpaux$ is in the interval of job \pat, then the statement is clear. 
Otherwise, consider a shorter augmenting path ending at a free slot $t\neq \tpaux$.
If the free slot $t$ is later than $\tpaux$, then one of the jobs on the short path is incident to a slot $t_1 \leq \tpaux$ and a slot $t_2 > \tpaux$, and, due to the interval constraint, also to the slot $\tpaux$.
Hence, the path does not become longer if we replace the sub-path between $t_2$ and $t$ by $\tpaux$, which is still an augmenting path.
If the target slot $t$ is earlier than $\tpaux$, then, by definition of the slot $\tpaux$, we have $t<\earlT[\ppat]$.
Then there is a job~\pat on the path with $\ass(j) > t$ and incident to slot $t$ or an earlier slot, i.e.~$\earlT \leq t$.
Upon arrival of job~\pat, the slot~$t$ was free as well. 
Hence, \FF would have assigned $j$ to the earlier free slot $t$.  
\end{proof}

If multiple shortest paths to the earliest possible slot exist, one is selected arbitrarily. 
We will discuss the influence of the choice of the path on the algorithm's performance at the end of Section~\ref{sec:ff}.

\subsection{Matching size under \FF}
We start our analysis of \FF with the primary objective in \probname{} -- the number of assigned jobs, i.e., the cardinality of the resulting matching.
In general, \FF can make unnecessary jobs rejections, as we will see later. 
On instances of Static \probname, however, \FF achieves a maximum number of assigned jobs.

\begin{remark}\label{lem:1-day}
	\FF yields a maximum matching on instances of Static \probname.
\end{remark}
In fact, as mentioned above, Static \probname is a special case of maximum online bipartite matching.  
 \FF, in turn, coincides with the SAP strategy for online bipartite matching complemented with a tie-breaking rule between paths of the same length: the earliest terminal slot is preferred.
	Hence, the optimality of \FF follows from the optimality of SAP \cite{Chaudhuri2009}.

To prove the next optimality result, we introduce the notion of a closed interval. 
\begin{definition}\label{def:closedInt}
    A \emph{closed interval} in an assignment $\ass\colon\patset\to\N\cup\{\bot\}$ for \probname is a time interval $I\subseteq \N$ without empty slots such that all \textpat{}s $\pat \in \ass^{-1}(I)$ have intervals $\patTW \subseteq I$. 
\end{definition}
Observe that any alternating path with respect to the matching $\ass$ in the underlying bipartite graph
that starts in a closed interval can visit only slots of this interval and jobs assigned to them.

\begin{theorem}\label{thm:FFonUniformTW}
	If all \textpat{}s in an instance of \probname have intervals of the same length, then \FF 
    yields a maximum matching.
\end{theorem}

 \begin{proof}
We show that the interval of each rejected job lies in a closed interval (see 
Definition~\ref{def:closedInt}), so that augmenting paths starting with the rejected job never reach an empty slot.

Consider a rejected \textpat $\ppat$ and the \FF assignment $\ass \colon \patset \to \N$ upon the arrival of $\ppat$.
As \textpat $\ppat$ is rejected, all slots in the interval $\patTW[\ppat]$ are occupied. 
	Let $\tpmax$ be the first slot after slot $\lateT[\ppat]$ which is still free:
	\[ \tpmax \coloneqq \min\left\{\tp > \lateT[\ppat] \mid {\ass^{-1}(\tp)} =\emptyset \right\}.
	\]
    We attempt to construct an augmenting path from \textpat $\ppat$ to slot $\tpmax$. If we succeed, we find a contradiction to the rejection of $\ppat$; otherwise, we show the existence of a closed interval containing the interval of $\ppat$.
    
    We iteratively construct a sequence of \textpat{}s as follows.
    First, denote $\pat_0\coloneqq \ppat$, and let \[
    \tp_1 \coloneqq \max\left\{\lateT \mid \pat \in \ass^{-1}(\patTW[\ppat]) \right\} \text{ and } \pat_1\coloneqq \argmax\{\lateT \mid \pat \in \ass^{-1}(\patTW[\ppat]) \}
    \]
    be the \textpat with the latest deadline among the \textpat{} assigned to the interval $\patTW[\ppat]$.
    In the following iterations $k$ for $k\in \N$, we set 
    \[
    \tp_{k} \coloneqq \max\left\{\lateT \mid \pat \in \ass^{-1}(\patTW[\pat_{k-1}])
    \right\}
    \text{ and }
    \pat_{k} \coloneqq \argmax\left\{\lateT\mid \pat \in \ass^{-1}(\patTW[\pat_{k-1}]) \right\}
    \]
    until either $\tp_{k} \geq \tpmax$ or $\tp_{k} = \tp_{k-1}$ occurs for some $k$.
    Note that by construction, $t_k \geq t_{k-1}$ is true for all $k\in \N$.

    If the process terminates with $\tp_{k} \geq \tpmax$, then the sequence \[Q\coloneqq(\ppat, \ass(\pat_1), \pat_1, \ass(\pat_2), \pat_2, \ldots, \pat_k, \tpmax)\] 
    is an {augmenting path} for \textpat $\ppat$ and assignment $\ass$.
	Indeed,
    slot $\ass(\pat_\ell)$ is feasible for \textpat $\pat_{\ell-1}$ by construction of the \textpat sequence, as $\ass(\pat_\ell) \in \patTW[\pat_{\ell-1}]$ for all $\ell \in \{1,\ldots,k-1\}$.
    The free slot $\tpmax$ is feasible for the last \textpat of the sequence $\pat_k$ due to the following relations:
\[\earlT[\pat_k] \leq \ass(\pat_{k}) \leq \lateT[\pat_{k-1}] = \tp_{k-1} < \tpmax.\] Finally, the slots in $Q$ are not fixed, since $\ass(\pat_\ell) \geq \earlT[\pat_{\ell-1}] \geq \earlT[\ppat]$ is true by construction for all $\ell \in \oneto{k}$.
	Although $Q$ is not necessarily the shortest augmenting path used by \FF, the existence of sequence $Q$ guarantees the existence of a shortest such path. 
	Hence, \textpat $\ppat$ would not be rejected by \FF.

	So suppose that the sequence construction stops at iteration $k$ with $t_k =  t_{k-1}$ and $t_k < \tpmax$.
    Here, we show that there exists a closed interval for $\ass$ containing the interval of $\ppat$.
	For all \textpat{}s~$\pat \in \ass^{-1}(\patTW[\pat_{k-1}])$ we have $\lateT[\pat] \leq t_k$; consequently, since all \textpat{}s have intervals of the same length, for any \textpat $\pat$ with $\ass(\pat) < \earlT[\pat_{k-1}]$ it also follows that $\lateT < \lateT[\pat_{k-1}] = t_k$.
    Hence, we set $t_k$ as the right bound of the sought closed interval.

    To determine the left bound of the closed interval, we consider the latest empty slot before \textpat
    $\ppat$'s interval:
    \[ \tpmin \coloneqq \max\left\{\tp < \earlT[\ppat] \mid {\ass^{-1}(\tp)} =\emptyset \right\}.
	\] 
	Observe that for any \textpat \pat assigned to a slot $\ass(\pat)>\tpmin$, the slot $\tpmin$ cannot be feasible -- otherwise, \textpat $\pat$ would have been assigned by \FF to the earlier slot $\tpmin$.
	Hence, for any \textpat $\pat$, from
	$\ass(\pat) > \tpmin \text{  follows that  } \earlT > \tpmin.$
	Combining the two bounds, we obtain 
	\[ \earlT \geq \tpmin+1 \quad\text{and}\quad \lateT \leq t_k \quad \text{for all}\ \pat \in \ass^{-1}([\tpmin+1, t_k]),
	\]
	that is, every slot of the interval $I\coloneqq [\tpmin+1, t_k]$ is occupied by a \textpat that requires a slot in this time interval.
 
    No alternating job-slot sequence with respect to $\ass$ starting with a slot in $I$ can ever leave interval $I$.
    Therefore, \textpat{}s assigned to $I$ in $\ass$ will not be replaced by \textpat{}s arriving later.
    Thus, also in the final \FF solution no offline alternating path starting in $I$, or particularly in $\patTW[\ppat]\subseteq I$, leads to an empty slot.
\end{proof}

\newcommand{\om}{\omega}
\newcommand{\twlen}{\delta}

\begin{theorem}
\label{lem:twotypesFF-UB}
\rev{On instances with at least two different interval lengths,} \FF is at most $\frac{2}{3}$-competitive for \probname with respect to \rev{matching cardinality}.
\end{theorem}
\begin{proof}
    We define two \textpat types: \emph{Urgent} \textpat{}s with preparation time of $1$ and interval length~$\twlen$, and \emph{regular} \textpat{}s with preparation time $2$ and intervals of length $\Delta \geq 3\delta-2$.
    Consider now the following instance $\Inst$ of \probname~(see Figure~\ref{fig:FF-LBreject}).
    First, $\ell\coloneqq \twlen-1$ regular \textpat{}s~$\pat_i$ for $i\in \{1,\ldots, \ell\}$ arrive at time $-1$
    (we extend the time scale to negative numbers to simplify the notation).
    Their interval is $[1, \Delta]$, and \FF assigns them to slots $\ass(i) = i$ for~$i\in \{1,\ldots, \ell\}$.
    Next, for each $i \in \{1,\ldots, \twlen\}$, an urgent \textpat $u_{i}$ with interval $[i, i+\ell]$ is revealed at time $i-1$. 
    Upon arrival of \textpat $u_i$, the first $\ell+i-1$ slots are occupied, so $u_i$ is assigned to slot $\ell+i$, the last slot of its interval.
    
    Finally, further $\ell$ urgent \textpat{}s $v_i$, $i\in \{1,\ldots, \ell\}$ with interval $[\twlen, \twlen+\ell]$ are revealed also at time~$\ell$.
    All slots in their interval are occupied. 
    The assignment up to the current time point~$\ell$ cannot be modified, so the regular \textpat{}s cannot be reassigned. 
    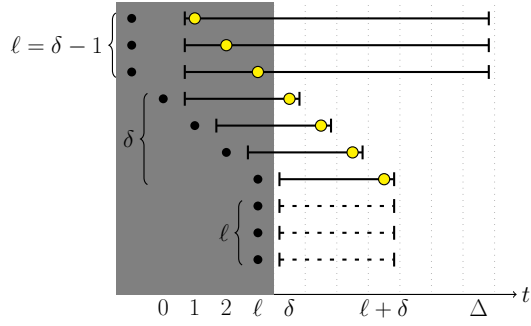
\begin{figure}[bt]
		\centering
		\resizebox{0.37\textwidth}{!}{

\begin{tikzpicture}[yscale=-0.85, scale=0.8]
	\def\k{3}
	\def\intclip{0.15}
	
	\tikzstyle{patint}= [
	line width=1.5pt, 
	|-|,
	draw]
	
	\tikzstyle{arrivenode}= [
		circle, 
		fill,
		minimum size=0.4mm,
		inner sep=0.8mm]
	\tikzstyle{assignmark}[black]= [
		fill=yellow,
		circle,
		thick,
		inner sep=1mm,
		draw=#1]
	
	\def\tmax{10}
	\draw[->, thick] (-1.98,3*\k+2.3) -- (\tmax+0.7,3*\k+2.3) node[right]{{\LARGE$t$}};
	
	\foreach \x in {-1,...,\tmax}{
		\draw[loosely dotted, gray] (\x, 0.5) -- ++ ($(0,3*\k+1.7)$);
	}
    \foreach \x/\t in {0/0, 1/1, 2/2, 3/\ell,4/\twlen, 10/\Delta, 7/\ell+\twlen}{
        \node[below] at (\x-0.5, 3*\k+2.3) {\LARGE$\t$};
    }

	\fill[gray, opacity=0.2] (-1.98, 0.4) rectangle (\k, 3*\k+2.3);
	
	\foreach \i in {1,...,\k}{
		\draw[patint]  ($(0, \i) + (\intclip,0)$) -- ($(10, \i) + (-\intclip, 0)$);
		\node[arrivenode] at (-1.5, \i) {};
		\node[assignmark] at ($(\i, \i) + (-0.5,0)$) {};
		
	}
	\foreach \i in {0,...,\k}{
		\draw[patint]  ($(\i, \k+\i+1) + (\intclip,0)$) -- ($(\i+\k+1, \k+\i+1) + (-\intclip, 0)$);
		\node[arrivenode] at (\i-0.5, \k+\i+1) {};
		\node[assignmark] at ($(\k+\i+1, \k+\i+1) + (-0.5,0)$) {};
		
	}
	\foreach \i in {1,...,\k}{
		\draw[patint, loosely dashed]  ($(\k, 2*\k+\i+1) + (\intclip,0)$) -- ($(\k+\k+1, 2*\k+\i+1) + (-\intclip, 0)$);
		\node[arrivenode] at (\k-0.5, 2*\k + \i+1) {};
	}

	\draw [thick, decorate, decoration = {brace,amplitude=5pt,raise=0mm}] (-2,\k+0.2) -- (-2,0.8) node[left=2mm, pos=0.5]{\LARGE$\ell=\twlen-1$};
	
    \draw [thick, decorate, decoration = {brace,amplitude=5pt,raise=0mm}] (-1, 2*\k+1.2) -- (-1,\k+0.8) node[left=2mm, pos=0.5]{\LARGE $\twlen$};
  
	\draw [thick, decorate, decoration = {brace,amplitude=5pt,raise=0mm}] (-1+\k, 3*\k+1.2) -- (-1+\k,2*\k+1.8) node[left=2mm, pos=0.5]{\LARGE $\ell$};
	
\end{tikzpicture}

		}
	\caption{Upper bound instance for $\twlen=4$, $\ell=3$. The arrival order of jobs is from top to bottom. The arrival times are marked with dots and the \FF assignment with yellow circles.
    The grayed-out part of the solution is fixed upon arrival of the last $\ell$ jobs. 
    The dashed jobs are rejected.
    }
	\label{fig:FF-LBreject}
    \end{figure}
	Hence, \textpat{}s $v_i$, $i\in \{1,\ldots, \ell\}$,  are all rejected by \FF, so \FF achieves the objective value $\text{FF}(\Inst) = 2\twlen-1$. 
	An optimal assignment $\beta$ assigns all \textpat{}s to a feasible slot:
    \begin{equation*} 
		\begin{array}{ll}
		    \beta(\pat_i)\, = \twlen+\ell+i,& 1\leq i \leq \ell, \\
		      \beta(u_i) = i,&1\leq i \leq \twlen,\\
            \beta(v_i) = \twlen +i, &1\leq i \leq \ell;
		\end{array}
	\end{equation*}
hence, $\text{OPT}(\Inst) = 3\twlen-2$.
From 
\[
\frac{\text{FF}(\Inst)}{\text{OPT}(\Inst)} = \frac{2\twlen-1}{3\twlen-2} \underset{\twlen\to\infty}{\longrightarrow} \frac{2}{3}
\]
follows that the competitive ratio of \FF is not greater than $\frac{2}{3}$.
\end{proof}

We complement this result with a tight lower bound. Even if we drop the interval constraints and allow for arbitrary sets of feasible slots, \FF is at least $\frac{2}{3}$-competitive.
\begin{theorem}\label{lem:FF-LB}
	\FF is at least $\frac{2}{3}$-competitive for \probname with respect to the number of assigned jobs.
\end{theorem}  
\begin{proof}
	\newcommand{\ffpats}{\patset_{\text{FF}}}
        \newcommand{\optpats}{\patset_{\text{OPT}}}
	\newcommand{\assff}{f}
	\newcommand{\assopt}{o}
	
	Consider an arbitrary instance of \probname.
	Let $\ffpats$ be the set of \textpat{}s assigned by \FF, and $\optpats$ be the set of \textpat{}s assigned in an optimal solution;
    let $\assff\colon \ffpats \to \N$ and $\assopt\colon \optpats \to \N$ be the corresponding assignments.
    \rev{Observe that the sets of jobs assigned by some assignment form a matroid.
    It is clearly an independence system; to see the that the exchange property holds, 
    consider two assignments on job sets $\patset$ and $\patset'$ with underlying matchings~$M$ and~$M'$, respectively, with $\abs{M}<\abs{M'}$.
    Their symmetric difference contains an augmenting path for~$M$~\cite[Proof of Theorem 10.7]{Korte_Vygen_2018}. 
    Augmenting $M$ along this path yields a new matching $M''$ that assigns all jobs of $\patset$ and one additional job from $\patset' \setminus \patset$.
    Hence, without loss of generality, we select matroid base $\optpats$ and a corresponding  assignment $\assopt$ so that $\ffpats \subseteq \optpats$.}
    
    To compare the cardinalities of the two assignments, we construct alternating job-slot sequences that help to  partition the job set. 
     For each \textpat~\mbox{$\pat_0 \in \optpats \setminus \ffpats$} rejected by \FF but assigned by the optimal algorithm, we construct a sequence $Q(\pat_0)$ by first adding job~$\pat_0$ and its optimal slot $s_0 = \assopt(\pat_0)$ to it. 
     Next, in each iteration~$k\in\N$, we consider the job~$\pat_k =\assff^{-1}(s_{k-1})$ that is assigned by \FF to the optimal slot~$s_{k-1}$ of the previous job in the sequence.
     \rev{By our assumption, every job assigned by $\assff$ also has a slot under assignment $\assopt$.}
    Job~$\pat_k$ and its optimal slot~$s_k=\assopt(\pat_k)$ are added to the sequence~$Q(\pat_0)$. 
    We repeat these iterations until \rev{we reach a slot~$s_k$ that is free in~$\assff$.}

    We show that \rev{ each 
    sequence} contains at least three \textpat{}s.
    If a sequence $Q(\pat_0)$ for some rejected \textpat $\pat_0$ contains only one \textpat, then $Q(\pat_0) = (\pat_0, \assopt(\pat_0))$, and slot $\assopt(\pat_0)$ is free in assignment $\assff$. So \FF would have assigned $\pat_0$ to this slot.

    Now consider \rev{a 
    sequence $Q(\pat_0) = (\pat_0, s_0, \pat_1, s_1)$  with two \textpat{}s, with} \mbox{$s_0 = \assopt(\pat_0)$}, $s_1 = \assopt(\pat_1)$ and $\assff(\pat_1) = s_0$ (see Figure~\ref{fig:FF-lb}).
    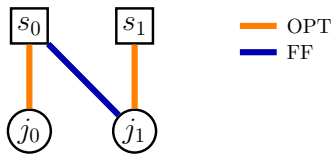
\begin{figure}[tb]
    \centering
    \resizebox{3.5cm}{!}{

	\begin{tikzpicture}[scale=0.9]
		
		\tikzstyle{patcirc}[black] = [
		draw=#1,
		circle,
		very thick,
		inner sep=1.5pt]
%
%
%
		\tikzstyle{slot} = [
			draw,
			very thick,
			rectangle,
			inner sep=1mm,
            minimum size = 6mm
		]
		
		\coordinate (cj1) at (0,0);
		\coordinate (cj2) at (2,0);
		\coordinate (cs1) at (0,2);
		\coordinate (cs2) at (2,2);
		
		\node[patcirc] (j1) at (cj1) {\Large{$j_0$}};
		\node[patcirc] (j2) at (cj2) {\Large$j_1$};
		\node[slot] (s1) at (cs1) {\Large$s_0$};
		\node[slot] (s2) at (cs2) {\Large$s_1$};
		
		\draw[orange, ultra thick, line width = 3pt] (j1.north) -- (s1.south);
		\draw[orange, line width = 3pt] (j2.north) -- (s2.south);
		
		\draw[cyan!50!black, line width = 3pt, shorten >=-0.5] (j2.north west) -- (s1.south east);

        \draw[orange, line width = 3pt] (4,2) -- (4.7,2) node[anchor=west, text=black]{\Large OPT};
        \draw[cyan!50!black, line width = 3pt] (4,1.3) -- (4.7,1.3) node[anchor=west, text=black]{\Large FF};
		
	\end{tikzpicture}
}
    \caption{An alternating sequence with two jobs.}\label{fig:FF-lb}
    \end{figure}
    Since slot $s_1$ is free in the \FF solution and \textpat $\pat_1$ is assigned to slot $s_0$, we have $s_0 < s_1$. 
    Furthermore, we argue that $\pat_1$ arrived before~$\pat_0$ and that upon arrival of $\pat_0$, job $\pat_1$ was already assigned to $s_0$.
    If this is the case, then the path $(\pat_0, s_0, \pat_1, s_1)$ has been an augmenting path for $\pat_0$, so \FF would 
    reassign $\pat_1$ instead of rejecting \textpat $\pat_0$. 
    This augmenting path is feasible  
    because all its edges are contained in the optimal assignment $\assopt$ or in the \FF assignment~$f$.

    So suppose that, upon arrival of \textpat $\pat_0$, slot $s_0$ was occupied not by $\pat_1$, but by some further \textpat $\pat_2$.
    In the final \FF assignment we have $\assff(\pat_1)=s_0$; hence, at some later time point, job $j_1$ is reassigned, or assigned straight upon its arrival, to $s_0$. 
    Slot $s_1$ is still available at that time; hence, the augmenting path (re-)assigning $\pat_1$ to $s_0$ can be shortened by assigning $\pat_1$ to a free slot $s_1$.
    This contradicts slot $s_1$ remaining empty in the final \FF assignment.
    So~$\pat_1$ was occupying slot $s_0$ upon arrival of $\pat_0$, and \FF would have assigned~$\pat_0$ by reassigning~$\pat_1$. 
    
    Next, we express the number of jobs assigned or rejected by both algorithms through the number of 
    \rev{sequences~$n_{Q}$.
    We have shown that each sequence contains at least three \textpat{}s, where the first job is in $\optpats\setminus\ffpats$ and the others in $\abs{\ffpats \cap \optpats}$.
    This yields a bound of~$n_Q \leq \frac{1}{2}\abs{\ffpats \cap \optpats}$. 
    As each \textpat{} in $\optpats\setminus\ffpats$ occurs at the start of a unique sequence, we have
    \[
        \abs{\optpats \setminus \ffpats} = n_Q \leq \frac{1}{2}\abs{\ffpats \cap \optpats}.
    \]
    Finally, we obtain
    \begin{align*}
        \abs{\optpats} &= \abs{\optpats \cap \ffpats} + \abs{\optpats\setminus\ffpats}\leq \frac{3}{2}\abs{\optpats \cap \ffpats} \leq \frac{3}{2}\abs{\ffpats}.
    \end{align*}
    }
\end{proof}


\subsection{Reassignments by \FF}
Another natural performance measure for online algorithms for \probname is the number of reassignments made by the algorithm. Before we present our analysis of \FF for this measure, note that we did not specify how ties between shortest augmenting paths from the current \textpat to the earliest free slot are resolved. Although each of the shortest paths yields the same number of reassignments in that iteration, we will see that the resulting solutions lead to different numbers of reassignments in later iterations.

All possible shortest augmenting paths can be uniquely represented and compared by the sequence of slots they visit on the way from the new \textpat to the target slot. One possible rule for path selection is to choose the lexicographically smallest sequence of slots. Intuitively, this rule favors early slots for incoming \textpat{}s, and the path can be found in polynomial time.

We show in Lemma~\ref{lem:LB-NlogN} that with this path choice rule, which we denote \emph{lex-min} for short, \FF makes at least $\Omega(n \log n)$ reassignments in the worst case. Observe that this bound coincides with the lower bound on the worst-case number of reassignments for the SAP strategy for online bipartite matching \cite{Grove1995}.

\begin{lemma}\label{lem:LB-NlogN}
	In the worst case, \FF makes $\Omega(n \log n)$ reassignments on instances with $n$ \textpat{}s, even on instances of Static \probname.
\end{lemma} 

\begin{proof}
Here, we provide only a sketch of the proof; the full proof is provided in~\Cref{apx:monsterProof}.

Consider an {upper-triangle} instance with $N = 2^n$ \textpat{}s arriving at time $0$.
Jobs are enumerated in the order of arrival, \textpat $i\in \oneto{N}$ has interval $[1, N+1-i]$, see Figure~\ref{fig:FFpyramid}.
	
To describe how \FF works on this instance, we divide the set of \textpat{}s into $n+1$ subsequent phases: phase $k\leq n$ contains $\frac{N}{2^{k}}$ \textpat{}s, and phase $n+1$ contains one \textpat.

In the first phase, the first half of all \textpat{}s arrives and is assigned by \FF to the first $\frac{N}{2}$ slots. In the second phase, the half of the remaining \textpat{}s arrives. As there is no free slot for any of these new \textpat{}s, each new \textpat causes a reassignment of one of the \textpat{}s of phase one.
According to the lex-min rule for choosing augmenting paths, \textpat{}s of phase $2$ are assigned to the first $\frac{N}{4}$ slots, and the first $\frac{N}{4}$ \textpat{}s of phase $1$ are reassigned to the third block of $\frac{N}{4}$ slots, i.e., to slots $\frac{N}{2}+1, \ldots, \frac{3N}{4}$. 
	
In each subsequent phase $k$, the half of already present \textpat{}s 
 gets reassigned. 
The chain of reassignments caused by \textpat $i$ ends with slot $i$, and 
each \textpat moves 
$\frac{N}{2^{k-1}}$ slots to the right,
so each chain contains $2^{k-1}-1$ reassignments (see Figure~\ref{fig:FFpyramid}). 
Hence, in phase $k$ there are $\frac{N}{2^k}(2^{k-1}-1) = \frac{N}{2} - \frac{N}{2^k}$ reassignments in total.
We refer to Appendix~\ref{apx:monsterProof} for detailed calculations.

\begin{figure}[tb]
	\centering
    \resizebox{0.45\textwidth}{!}{
\pgfdeclarelayer{bg}
\pgfsetlayers{bg, main}	

\begin{tikzpicture}[yscale=-1, scale=0.8]
	\def\k{3}
	\def\intclip{0.15}
	
	\tikzstyle{patint}= [
	line width=1.5pt, 
	|-|,
	draw]
	
	\tikzstyle{arrivenode}= [
	circle, 
	fill,
	minimum size=0.4mm,
	inner sep=0.8mm]
	\tikzstyle{assignmark}[black]= [
	fill=yellow,
	circle,
	thick,
	inner sep=1mm,
	draw=#1]
	
	\def\tmax{17}
	\draw[->, thick] (-0.7,\tmax-1.5-0.2) -- (\tmax+0.7,\tmax-1.5-0.2) node[right]{{\LARGE$t$}};
	
	\foreach \x in {0,...,\tmax}{
		\draw[loosely dotted, gray] (\x, 0.2) -- ++ ($(0, \tmax - 1.9)$);
	}
	
	\foreach \a [count=\i] in {9, 10, 11, 12, 5, 6, 7, 8, 1, 2, 3, 4}{
		\draw[patint]  ($(0, \i) + (\intclip,0)$) -- ($(17-\i, \i) + (-\intclip, 0)$);
		\node[assignmark] at ($(\a, \i) + (-0.5,0)$) {};		
	}
	\foreach \i in {13, 14}{
		\draw[patint, dashed]   ($(0, \i) + (\intclip,0)$) -- ($(17-\i, \i) + (-\intclip, 0)$);
		\node[assignmark, fill=white] at ($(\i-12, \i) + (-0.5,0)$) {};
	}

	\draw [thick, decorate, decoration = {brace,amplitude=7pt,raise=0mm}] (-0.5,8+0.2) -- (-0.5,0.8) node[left=4mm, pos=0.5]{\LARGE{Phase 1}};
	\draw [thick, decorate, decoration = {brace,amplitude=7pt,raise=0mm}] (-0.5,12+0.2) -- (-0.5,8+0.8) node[left=4mm, pos=0.5]{\LARGE{Phase 2}};
  	\draw [thick, decorate, decoration = {brace,amplitude=5pt,raise=0mm}] (-0.5,14+0.2) -- (-0.5,12+0.8) node[left=4mm, pos=0.5]{\LARGE{Phase 3}};
	
    \begin{pgfonlayer}{bg}
    \foreach \x/\y in {9/1, 5/5, 1/9}{
		\fill[cyan, opacity=0.3, rounded corners] (\x-0.9,\y-0.4) rectangle ++(0.9, 0.9);
        \fill[orange, opacity=0.6, rounded corners] (\x-0.9 +0.9,\y-0.4 +0.9) rectangle ++(0.9, 0.9);
		\draw[cyan!80!black, ultra thick, rounded corners] (\x-0.9 + 4,\y-0.4) rectangle ++(0.9, 0.9);
        \draw[orange!60!black, ultra thick, rounded corners] (\x-0.9 + 4 +0.9,\y-0.4 +0.9) rectangle ++(0.9, 0.9);
		\draw[->,  ultra thick, cyan!70!black] (\x+1.5, \y+0.5) -- ++(1.0,0);
	}
    \draw[cyan!80!black, ultra thick, rounded corners] (1-0.9, 13-0.4) rectangle ++(0.9, 0.9);
    \draw[orange!60!black, ultra thick, rounded corners] (1-0.9 +0.9, 13-0.4 +0.9) rectangle ++(0.9, 0.9);
    \end{pgfonlayer}
 
\end{tikzpicture}

    }
\caption{Example instance with the reassignments in phase $3$. Under lex-min rule, \textpat{}s of the current phase occupy the first slots in their order of arrival, {and all other jobs are assigned in parallel diagonal blocks}. The reassignment chains propagate {this block structure} through the instance.}\label{fig:FFpyramid}
\end{figure}
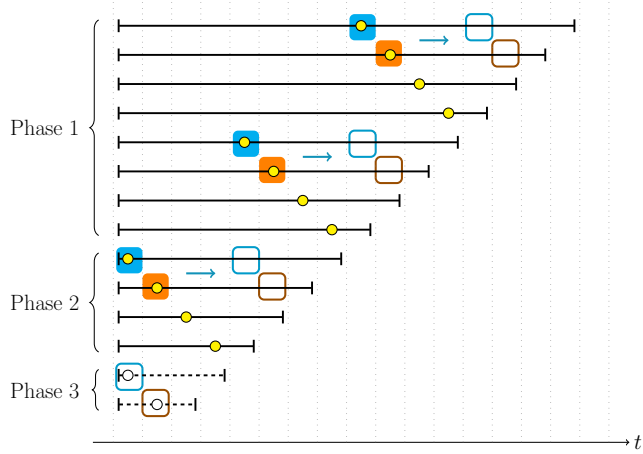

Hence, the total number of reassignments up to phase $\log N$ is

\begin{align*}
\sum_{k=2}^{\log N} \left(\frac{N}{2} - \frac{N}{2^k}\right) &= \frac{N}{2}(\log(N)-1) - N\sum_{k=2}^{\log N}\frac{1}{2^k}\\&= \frac{N}{2}(\log(N)-1) - N\left(\frac{1}{2} - \frac{1}{N}\right) \in \Omega(N \log N)\;.
\end{align*}
\end{proof}

Observe that on the \probname instance in Lemma~\ref{lem:LB-NlogN}, a different path choice rule, e.g., the \emph{lex-max} rule that chooses the lexicographically maximal sequence of visited slots, would lead to only $\frac{N}{2}$ reassignments in total.

However, the lex-max rule is also sub-optimal: Consider an instance similar to that in Figure~\ref{fig:FFpyramid}, but in which the second half of the \textpat{}s have right-aligned intervals $[i-\frac{N}{2}, \frac{N}{2}]$ for~$i\in\{\frac{N}{2}+1,\ldots,N\}$.
By reasoning analogously to the proof of Lemma~\ref{lem:LB-NlogN}, the lex-max path choice rule yields $\O(n\log n)$ reassignments on this instance, whereas lex-min requires only $\frac{N}{2}$ reassignments.


As for the upper bound on the number of reassignments, there is a known upper bound of $\O(n \log^2 n)$ reassignments for \obm under the Shortest Augmenting Path strategy. 
This bound holds in particular for Static \probname as a special case of \obm, but not necessarily for \probname in general.
Unfortunately, we were unable to show an analogous upper bound for \probname. 
Our problem differs from \obm in two aspects: The interval structure of feasible sets and the temporal component that fixes parts of the graph during the online arrival process. 
While the interval constraints possibly reduce the complexity of the problem by restricting possible adversarial instances, 
the temporal aspect of \probname instances makes any SAP-based algorithm less potent, since there are fewer augmenting paths available, as Figure~\ref{fig:exmp} demonstrates.

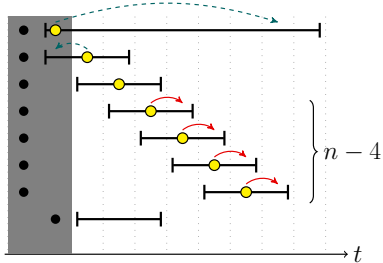
\begin{figure}[hbt]
	\centering
    \resizebox{!}{3.3cm}{
\definecolor{combi-darkcyan}{RGB}{0,100,100}
	
	\begin{tikzpicture}[yscale=-0.85, scale=0.76]
		\def\k{6}
		\def\l{2}
		\def\intclip{0.15}
		\def\tmax{10}

		\tikzstyle{patint}= [
		line width=1.5pt, 
		|-|,
		draw]
		
		\tikzstyle{arrivenode}= [
		circle, 
		fill,
		minimum size=0.4mm,
		inner sep=0.8mm]
		\tikzstyle{assignmark}[black]= [
		fill=yellow,
		circle,
		thick,
		inner sep=0.9mm,
		draw=#1]

		\draw[->, thick] (-0.1,\k+2.3) -- (\tmax+0.7,\k+2.3) node[right]{{\LARGE$t$}};
		
		\foreach \x in {0,...,\tmax}{
			\draw[loosely dotted, gray] (\x, -0.5) -- ++ ($(0,\k+1+1.7)$);
		}
				
		\fill[gray, opacity=0.2] (0.01, -0.5) rectangle (2, \k+2.3);
		
		\draw[patint]  ($(1, 0) + (\intclip,0)$) -- ($(\l+\k+2, 0) + (-\intclip, 0)$);
		\node[arrivenode] at (0.5, 0) {};
		\node[assignmark] at ($(1+1, 0) + (-0.5,0)$) {};
		
		\foreach \i in {1,...,\k}{
			\draw[patint]  ($(\i, \i) + (\intclip,0)$) -- ($(\i+\l+1, \i) + (-\intclip, 0)$);
			\node[arrivenode] at (0.5, \i) {};
			\node[assignmark] at ($(\l+\i, \i) + (-0.5,0)$) {};
		}
		\foreach \i in {3,...,\k}{
			\draw[->,>=stealth',thick, red!90!black] ($(\i+2, \i) + (-0.5,0) + (0, -0.3)$) to[out=-45,in=-135] ($(\i+\l, \i) + (0.5,0) + (0, -0.3)$);
		}
	
		\draw[->,>=stealth',thick, combi-darkcyan, dashed] ($(1, 0) + (0.5,0) + (0, -0.3)$) to[out=-20,in=-160] ($(\l+\k, 0) + (0.5,0) + (0, -0.3)$);
		\draw[->,>=stealth',thick, combi-darkcyan, dashed] ($(2, 1) + (0.5,0) + (0, -0.3)$) to[out=-135,in=-45] ($(1, 1) + (0.5,0) + (0, -0.3)$);
		
		\draw[patint]  ($(2, \k+1) + (\intclip,0)$) -- ($(\l+3, \k+1) + (-\intclip, 0)$);
		\node[arrivenode] at (1.5, \k+1) {};

		\draw [thick, decorate, decoration = {brace,amplitude=5pt,raise=0mm}] (\tmax-0.5,3-0.4) -- (\tmax-0.5,\k+0.4)  node[right=2mm, pos=0.5]{\LARGE$n-4$};

	\end{tikzpicture}
	
    }
\caption{An instance with $n=8$ jobs with the \FF assignment, identical to an SAP assignment, upon arrival of the last job. The reassignments by \FF in \probname are shown in red. Online bipartite matching admits a shorter augmenting path of length~$2$, shown in blue dashed.}\label{fig:exmp}
\end{figure}

Finally, we again consider a special case of instances with uniform interval length, assuming that these instances allow for better bounds also with respect to the number of reassignments.



\begin{lemma}\label{lem:FFreassUB-uniform}
    On instances of \probname with uniform interval length $\twlen$, \FF makes $\Theta(\twlen n)$ reassignments in the worst case.
\end{lemma}
 \begin{proof}
    We first show that on the described instances, \FF never reassigns a \textpat to an earlier slot.
    To begin with, the properties of \FF imply that the target slot at the end of any reassignment sequence is later than the start slot of the sequence, which lies within the new \textpat's interval.
    Hence, any reassignment sequence of \FF ends with a \emph{forward reassignment} -- a reassignment to a later slot.

    Next, suppose that some reassignment sequence with respect to the start assignment $\ass$ contains a \emph{backward reassignment}
    \[
        \pat \to \tp \text{ with }\tp < \ass(\pat), \text{ where } t= \ass^{-1}(\pat'),
    \]
    followed by a forward reassignment $\pat' \to \tp'$, i.e., $\tp' > \tp$.
    
    If $\tp' < \ass(\pat)$ and thus $\tp' \in [\tp, \ass(\pat)]$, then slot $\tp'$ is feasible for \textpat $\pat$, so the reassignment sequence can be shortened by assigning $\pat \to \tp'$ directly.
    If $\tp' > \ass(\pat)$, then the slots $\tp$ and $\tp'$ are less than one interval length apart, since they are both feasible for \textpat $\pat'$. 
    Now consider the \textpat $\pat_0$ that precedes $\pat$ in the sequence and is reassigned to the slot~$\ass(\pat)$ (\textpat $\pat_0$ can also be the new \textpat that initiates the reassignments).
    Since its interval contains slot $\ass(\pat)$, it also contains at least one of the slots $\tp$ and $\tp'$.
    In both cases, the reassignment sequence can be shortened by assigning \textpat $\pat_0$ directly to the slot $\tp$ or $\tp'$.

    Consequently, \FF makes only forward reassignments on instances with uniform interval length $\twlen$. 
    Every \textpat can thus be reassigned at most $\twlen-1$ times, which yields the overall upper bound of $(\twlen-1)n = \O(\twlen n)$ reassignments for instances with $n$ \textpat{}s. 
    If the interval length $\twlen$ is considered as fixed, we obtain the bound of $\O(n)$ reassignments.

    Next, we complement the result by showing an lower bound of $\Omega(n)$ reassignments.
    Consider the following instance with interval length $\twlen$ (see Figure~\ref{fig:uniformStaircase}).
    First, $\twlen-1$~identical \textpat{}s with intervals $[1, \twlen]$ arrive and are assigned to the first $\twlen-1$ slots. 
    Next, $m$~\textpat{}s with intervals $[i+1, i+\twlen]$ for $i\in \oneto{m}$ arrive and become assigned each to its
    one but last slot. The first $m+\twlen-1$ slots are now occupied.
    Now one further \textpat with interval $[1,\twlen]$ arrives.
    Each of the $m$ \textpat{}s with later intervals is shifted by one slot forward, yielding $m$ reassignments for an instance with $n=m+\twlen$ \textpat{}s.
    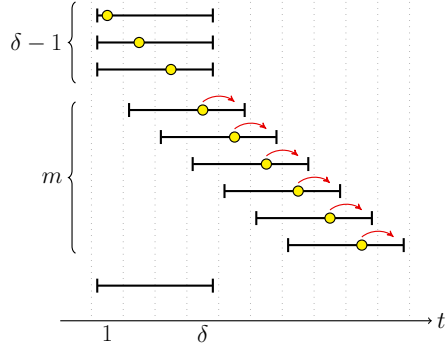
\begin{figure}[htb]
    \centering
    \resizebox{0.35\textwidth}{!}{
\definecolor{combi-darkcyan}{RGB}{0,100,100}

\begin{tikzpicture}[yscale=-0.85, scale=0.8]
	\def\k{6}
	\def\l{3}
	\def\intclip{0.15}
	\def\tmax{10}

	\tikzstyle{patint}= [
	line width=1.5pt, 
	|-|,
	draw]
	
	\tikzstyle{arrivenode}= [
	circle, 
	fill,
	minimum size=0.4mm,
	inner sep=0.8mm]
	\tikzstyle{assignmark}[black]= [
	fill=yellow,
	circle,
	thick,
	inner sep=0.9mm,
	draw=#1]

	\draw[->, thick] (-0.98,\k+\l+1+2.3) -- (\tmax+0.7,\k+\l+1+2.3) node[right]{{\LARGE$t$}};
	
	\foreach \x in {0,...,\tmax}{
		\draw[loosely dotted, gray] (\x, 0.5) -- ++ ($(0,\l+\k+1+1.7)$);
	}
	\foreach \x/\t in {1/1, 4/\twlen}{
		\node[below] at (\x-0.5, \l+\k+1+2.3) {\LARGE$\t$};
	}
	
	\foreach \i in {1,...,\l}{
		\draw[patint]  ($(0, \i) + (\intclip,0)$) -- ($(\l+1, \i) + (-\intclip, 0)$);
		\node[assignmark] at ($(\i, \i) + (-0.5,0)$) {};
	}

	\foreach \i in {1,...,\k}{
		\draw[patint]  ($(\i, \l+\i+0.5) + (\intclip,0)$) -- ($(\i+\l+1, \l+\i+0.5) + (-\intclip, 0)$);
		\node[assignmark] at ($(\l+\i, \l+\i+0.5) + (-0.5,0)$) {};
  
		\draw[->,>=stealth',thick, red!90!black] ($(\l+\i, \l+\i+0.5) + (-0.5,0) + (0, -0.3)$) to[out=-45,in=-135] ($(\l+\i, \l+\i+0.5) + (0.5,0) + (0, -0.3)$);
	}

	\draw[patint]  ($(0, \l+\k+2) + (\intclip,0)$) -- ($(\l+1, \l+\k+2) + (-\intclip, 0)$);
	\draw [thick, decorate, decoration = {brace,amplitude=5pt,raise=0mm}] (-0.5,\l+\k+1-0.2) -- (-0.5,\l+1+0.2) node[left=2mm, pos=0.5]{\LARGE$m$};
    \draw [thick, decorate, decoration = {brace,amplitude=5pt,raise=0mm}] (-0.5,\l+0.5) -- (-0.5, 1-0.5) node[left=2mm, pos=0.5]{\LARGE$\delta-1$};
	
\end{tikzpicture}

    }
    \caption{An example instance with uniform jobs at the moment of arrival of the last job. Yellow circles show the assignment at that moment, arrows indicate caused reassignments.}\label{fig:uniformStaircase}
    \end{figure}
\end{proof}
    
\subsection{Limited number of reassignments}\label{sec:k-ff}
Next, we introduce a modification of \FF which makes a limited number of reassignments per \textpat arrival, analogously to the modified SAP strategy considered by Shin et al.~\cite{Shin2020}.
 \emph{\FF with path limit $k$}, or $k$-\FF for short (Alg.~\ref{alg:k-FF}), is only allowed to use augmenting paths whose length does not exceed $2k+1$ for a fixed $k\in\N$, that is, paths which contain at most $k$ already scheduled \textpat{}s. Thus, k-\FF reassigns at most a constant number of \textpat{}s per new assignment.

\begin{algorithm}
\caption{\FF with path limit $k$ ({$k$-\FF})}\label{alg:k-FF}
\KwIn{\textpat $\ppat = (\req[\ppat], \earlT[\ppat], \lateT[\ppat])$, current slot assignment $\ass\colon\patset\to\N$}
\KwOut{new assignment $\ass'\colon\patset\cup\{\ppat\} \to\N\cup\{\bot\}$}
    \vspace*{1mm}
    $\tpaux \coloneqq \min\{\tp \geq \earlT[\ppat] \mid \ass^{-1}(\tp) = \emptyset\}$\;
    Find a shortest augmenting $\ppat$-$\tpaux$ path $P$ in $G_{\ass, \ppat}$\;
    \eIf{$\abs{P} \leq 2k+1$}{
    $\ass'\coloneqq \ass \,\Delta\, \left\{(\pat,\tp) \mid (\pat,\tp) \in P \text{ or } (\tp,\pat) \in P\right\}$\;
    }{
    $\ass' \coloneqq \ass \cup (\ppat, \bot)$\;
    }
    \KwRet{$\ass'$}
\end{algorithm}

Clearly, $k$-\FF makes at most $kn = \O(n)$ reassignments. 
Interestingly, both~\Cref{lem:FF-LB} and the example in~\Cref{lem:twotypesFF-UB} also hold for \FF with path limit. 
To see why~\Cref{lem:FF-LB} holds for $k$-\FF, recall that the only nontrivial part of the proof is showing that alternating sequences of length $4$ do not exist. 
The argument is based on two main statements. First, a path encoded by an alternating sequence $(\pat_0, s_0,\pat_1,s_1)$ is an augmenting path for the presumably rejected job $\pat_0$. Since the augmenting path uses only one reassignment, it is feasible for $k$-\FF with any $k\in \N$.
Second, the non-terminal job $\pat_1$ of the sequence was assigned to slot $s_0$ already upon arrival of $\pat_0$, as otherwise the path used to later reassign $\pat_1$ would not be the shortest. This argument clearly holds for $k$-\FF as well.

Hence, we obtain an assignment strategy which is bounded in both the number of rejections and reassignments.

\begin{corollary}
	The $k$-\FF assignment strategy for \probname is $\frac{2}{3}$-competitive with respect to the number of assigned jobs and makes at most $k\cdot n$ reassignments.
\end{corollary}

\blue{Since the competitive ratios of \FF and $k$-\FF for any $k\in \N$ are equal, 
there is no clear trade-off between the number of reassignments and the number of accepted jobs.} 
Nevertheless, with a limit on the number of reassignments per new assignment, we expect to obtain worse solutions on average. 
For instance, for any $k\in\N$, there are instances on which $(k+1)$-\FF yields a maximum matching, while $k$-\FF assigns only a \mbox{$\frac{k+1}{k+2}$-fraction} of \textpat{}s, cf.~Appendix~\ref{apx:k-FF}. 
In particular, $k$-\FF is not optimal for Static \probname.
Therefore, although $1$-\FF has the best worst-case guarantee on the number of reassignments, choosing a bigger $k$ or using unlimited \FF might still be preferable. 

\section{Proactive reassignments: Earliest Deadline First}\label{sec:edf}

So far, we have assumed that reassignments are made only if it is otherwise impossible to assign a newly arriving \textpat. 
Now we investigate an assignment algorithm that strategically reassigns \textpat{}s upon arrival of new requests. 

The \emph{\edf} algorithm (Alg.~\ref{alg:EDF}) mimics the optimal offline assignment strategy Earliest Deadline First. 
It maintains a maximal matching in the future part of the schedule
and rejects jobs that cannot be matched in addition to the already scheduled ones.
Ties in the ordering of jobs in the schedule are broken according to non-decreasing deadlines.

\rev{In contrast to the machine scheduling setting,} where jobs with earlier deadlines are prioritized and may push out other jobs from the schedule, our \edf for \probname respects the arrival order of jobs and rejects later jobs with an earlier deadline, if they cannot be assigned without discarding some already accepted jobs. 

\renewcommand{\ppat}{\pat'}
\begin{algorithm}
\caption{\edf}\label{alg:EDF}
\DontPrintSemicolon
\KwIn{\textpat $\ppat = (\req[\ppat], \earlT[\ppat], \lateT[\ppat])$, current slot assignment $\ass\colon\patset\to\N$}
\KwOut{new assignment $\ass'\colon\patset\cup\{\ppat\} \to\N\cup\{\bot\}$}
    \vspace*{1mm}
    $S \coloneqq \{\ppat\}\cup\{ \pat \in \patset \mid \ass(\pat) > \req[\ppat] \text{ and } \lateT[\pat] > \lateT[\ppat] \}$;\;
    sort $S$ by deadlines;\;
    initialize $\ass'\colon \pat\mapsto \begin{cases}
        \bot, & \pat \in S,\\
        \ass(j), & \pat \in \patset \setminus S;
    \end{cases}$\;
    \ForEach{$\pat \in S$}{
    \eIf{no free slot in $\patTW$}{
        $\ass' \coloneqq \ass \cup (\ppat, \bot)$;\tcc*{\small{reject $\ppat$, keep old assignment}}
        \KwRet{$\ass'$};
    }{
    $\ass'(\pat) \coloneqq \min\{\tp\in\patTW \mid (\ass')^{-1}(\tp) = \emptyset\}$;\tcc*[r]{\small{earliest free slot}}
    }
    }
    \KwRet{$\ass'$}
\end{algorithm}

\begin{theorem}\label{thm:edf}
	\edf maximizes the number of assigned jobs.
\end{theorem}
 \begin{proof}
    \newcommand{\patx}{z}
    \newcommand{\auxpat}{y}
    
    We show that a final \edf solution admits no offline augmenting paths. To this end, we show for each rejected job that its interval is contained in a closed interval.

    Let $\ppat$ be a \textpat rejected by \edf, and let $\ass \colon \patset \to \N$ be the slot assignment upon its arrival. 
    We first show that a closed interval for assignment $\ass$ existed at that time point, and argue afterwards that this interval remains in the final solution.

    As \textpat $\ppat$ was rejected, some \textpat $\patx \in \patset$ did not obtain a slot during the reassignments triggered by $\ppat$.
	Let $\assx \colon \overline{\patset} \to \N$, where $\overline{\patset} \subseteq \patset \cup \{\ppat\}$, denote the auxiliary slot assignment at the iteration in which $\patx$ is considered: $\assx$ is identical to $\ass$ on the slots up to the current slot $\tp_0 \coloneqq\req[\ppat]$ and on jobs with deadlines up to $\lateT[\ppat]$, 
    and assigns the remaining \textpat{}s of set $\overline{\patset}$
    according to the EDF rule.
	Note that either $\ppat = \patx$ or $\lateT[\ppat] < \lateT[\patx]$ is true; that is, the newly arrived \textpat either has no slot or has pushed out some other \textpat with a later deadline.
	
 We show that there is a closed interval for assignment $\assx$ that contains the interval $\patTW[\patx]$ of the pushed-out \textpat.
 To find the left bound of the interval, we proceed similarly to the proof of Theorem~\ref{thm:FFonUniformTW}: we iteratively construct a closure of the interval $[l,u] \coloneqq \patTW[\patx]$. In each iteration, we extend the interval towards the earliest release time of the \textpat{}s assigned to the current interval.
	Formally, let $r_0 \coloneqq \earlT[\patx]$.
	For $i \in \N$ define
	\[
		r_i \coloneqq \min\{\earlT \mid \pat \in \assx^{-1}([r_{i-1}, u])\} \qquad \text{and} \qquad \pat_i\coloneqq \argmin\{\earlT \mid \pat \in \assx^{-1}([r_{i-1}, u])\}.
	\]
	Observe that $0\leq r_i \leq r_{i-1}$ for all $i\in \N$, that is, the sequence is monotonically decreasing and bounded from below. 
	Hence, there exists $k \in \N$ for which $r_{k+1} = r_{k}$ occurs for the first time. 
	By construction, for each \textpat $\pat\in \assx^{-1}\left([r_k, u]\right)$ holds that $\earlT \geq r_{k+1} = r_k$.
	
    We now show that for these \textpat{}s also $\lateT \leq u$ is true.
	For any \textpat \pat assigned by $\assx$, the EDF order implies the following:
	for any further \textpat $\auxpat$
    with $\assx(\auxpat) \in \big[\earlT,  \assx(\pat)\big)$ we have $\lateT[\auxpat] \leq \lateT$.
	Otherwise, \pat would have been assigned to the earlier slot $\assx(\auxpat)$, since, at the moment when $\auxpat$ got assigned, both \textpat{}s \pat and $\auxpat$ were present and not yet fixed.
 
	We use this auxiliary statement to prove the upper bound on the deadlines of \textpat{}s assigned to the interval $[r_k, u]$ by induction.
	For any $\pat \in \assx^{-1}([l,u])$ we have $\lateT \leq u$ due to the EDF order.
	Now, we show that if $\lateT\leq u$ holds for any \textpat \pat in $\assx^{-1}\left([r_{i-1}, u]\right)$ for some $i\in \N$, then it also holds for all \textpat{}s in $\assx^{-1}\left([r_{i}, u]\right)$.
    Consider \textpat $\pat_{i}$ with release time $r_{i}$. 
	Since~$\assx(\pat_i) \in [r_{i-1}, u]$, we have $\lateT[\pat_i] \leq u$.
    For all $\pat$ with $\assx(\pat) \in \left[r_i, r_{i-1}\right] \subseteq \big[r_{i}, \assx(\pat_{i})\big)$
	we have $\lateT \leq \lateT[\pat_{i}] \leq u$ by the auxiliary statement above.
	Since the induction hypothesis holds for~$[r_{i-1}, u]$, we conclude that the hypothesis holds for the entire interval $[r_i, u]$.
	By induction, any \textpat assigned by \edf to the interval $[r_k, u]$ has a deadline $\lateT \leq u$.
	
	To sum up, any \textpat assigned to the interval $I\coloneqq [r_k, u]$ during the reassignment action of \textpat~$\ppat$ has $\earlT \geq r_k$ and $\lateT \leq u$. 
	To show that interval $I$ is closed under $\assx$, it remains to show that it is fully occupied. 
	
	Assume that there is an unoccupied slot $\tp_u \in [r_k, u]$; let $\tp_u$ be the latest such slot.
Clearly, $\tp_u < \earlT[\patx] = r_0$.
Since the sequence $(r_i)_{i=0}^{k}$ is strictly monotonously decreasing, there exists a unique $i\in \N$ such that $\tp_u \in [r_i, r_{i-1})$.
Then consider \textpat $\pat_i$. 
The empty slot $\tp_u$ is feasible for~$\pat_i$, and $\pat_i$ is assigned to a later slot $\assx(\pat_i) \geq r_{i-1} > \tp_u$, which contradicts 
the \edf assignment rule.
Consequently, all slots in the interval $[r_k, u]$ are occupied.

Hence, interval $I$ is closed under the auxiliary assignment~$\assx$.
Next, we show that is it also closed for assignment $\ass$.
If $\ppat = \patx$, then $\ass^{-1}(I) = \assx^{-1}(I)$,
as all jobs assigned by $\assx$, except for $\ppat$, are also part of assignment $\ass$ and can be assigned only to slots in $I$.
If $\ppat \neq \patx$, then $\ass^{-1}(I) = \assx^{-1}(I) \setminus \{\ppat\} \cup \{\patx\}$, and interval $I$ is still closed under $\ass$, since $\patTW[\patx] \subseteq I$.
Moreover, we have $\assx(\ppat) \in I$ (otherwise $\ass$ would be assigning $\abs{I}+1$ \textpat{}s to interval $I$), and hence $\patTW[\ppat] \subseteq I$.

It remains to show that interval $I$ stays closed in the final \edf assignment.
Observe that \edf makes only forward reassignments.
None of the jobs assigned to $I$ upon arrival of $\ppat$ can be moved forward outside $I$ due to their deadlines. 
Therefore, none of the jobs arriving after $\ppat$ will be able to push out some job from $I$.
Hence, the same set of \textpat{}s remains assigned to interval $I$ in the final solution, and so $I$ is again closed. 
Since the interval of the rejected \textpat $\ppat$ lies in a closed interval, there is no augmenting path for the final matching connecting $\ppat$ and an empty slot.
\end{proof}

Clearly, \edf is inefficient in terms of the number of reassignments.
We show that it achieves the trivial upper bound of $\Omega(n^2)$ reassignments per $n$ \textpat{}s.

\begin{lemma}\label{lem:EDF-n2}
    \edf makes $\Omega(n^2)$ reassignments in the worst case, even on instances of Static \probname.
\end{lemma}
\begin{proof}
	Consider again an instance consisting of $n$ \textpat{}s with monotonously decreasing deadlines, as in the proof of \Cref{lem:FF-LB} (see~\Cref{fig:FFpyramid}).
	Upon each \textpat arrival, all already present \textpat{}s will be reassigned one slot further, in order to assign the new \textpat to the first slot.
	Hence, the algorithm makes $\sum_{i=1}^{n} (i-1) = \frac{n(n-1)}{2}$ reassignments in total. 
\end{proof}

To complete the comparison of \FF and \edf, we consider the number of reassignments of \edf on instances with uniform interval lengths.
\begin{lemma}\label{lem:edf-uniform}
    If all $n$ \textpat{}s have the same interval length $\twlen$, then \edf makes $\Theta(\twlen n)$ reassignments. 
\end{lemma}
%
 \begin{proof}
By definition, \edf makes only forward reassignments, i.e., reassignments to later slots. 
Hence, we obtain the bound of $\O(\twlen n)$ reassignments analogously to the proof of Lemma~\ref{lem:FFreassUB-uniform}.

For the lower bound, consider the instance shown in Figure~\ref{fig:edf-uniform}.
The instance contains $n=\twlen+m$ jobs in total. 
The arrival of each of the last $\twlen-1$ jobs with earliest deadlines causes a reassignment of each of the first $m$ jobs by one slot,
which yields \[(\twlen-1)\cdot m = (\twlen-1)(n-\twlen) = \Omega(\twlen n)\] reassignments.
\begin{figure}[htb]
\centering
\resizebox{0.4\textwidth}{!}{
\definecolor{combi-darkcyan}{RGB}{0,100,100}

\begin{tikzpicture}[yscale=-0.85, scale=0.8]
	\def\k{6}
	\def\l{3}
	\def\intclip{0.15}
	\def\tmax{10}

	\tikzstyle{patint}= [
	line width=1.5pt, 
	|-|,
	draw]
	
	\tikzstyle{arrivenode}= [
	circle, 
	fill,
	minimum size=0.4mm,
	inner sep=0.8mm]
	\tikzstyle{assignmark}[black]= [
	fill=yellow,
	circle,
	thick,
	inner sep=0.9mm,
	draw=#1]

	\draw[->, thick] (-0.98,\k+\l+1+0.3) -- (\tmax+0.7,\k+\l+1+0.3) node[right]{{\LARGE$t$}};
	
	\foreach \x in {0,...,\tmax}{
		\draw[loosely dotted, gray] (\x, 0.5) -- ++ ($(0,\l+\k+1)$);
	}
	\foreach \x/\t in {1/1, 4/\twlen}{
		\node[below] at (\x-0.5, \l+\k+1+0.4) {\LARGE$\t$};
	}
	
	\foreach \i in {1,...,\l}{
		\draw[patint, dashed]  ($(0, \i+\k) + (\intclip,0)$) -- ($(\l+1, \i+\k) + (-\intclip, 0)$);
		\node[assignmark, fill=none] at ($(\i+1, \i+\k) + (-0.5,0)$) {};
	}

	\foreach \i in {1,...,5}{
		\draw[patint]  ($(\k-\i, \i) + (\intclip,0)$) -- ($(\k-\i+\l+1, \i) + (-\intclip, 0)$);
		\node[assignmark] at ($(\k-\i+1, \i) + (-0.5,0)$) {};
        \node[assignmark, fill=none] at ($(\k-\i+1+\l, \i) + (-0.5,0)$) {};
  
		\draw[->,>=stealth',thick, red!90!black] ($(\k-\i+1, \i-0.05) + (-0.5,0) + (0, -0.3)$) to[out=-45,in=-135] ($(\k-\i+1, \i-0.05) + (0.5,0) + (0, -0.3)$);
	}
 	\draw[patint]  ($(0, \k) + (\intclip,0)$) -- ($(\l+1, \k) + (-\intclip, 0)$);
	\node[assignmark] at ($(1, \k) + (-0.5,0)$) {};

 
	\draw [thick, decorate, decoration = {brace,amplitude=5pt,raise=0mm}] (-0.5,\k-0.8) -- (-0.5,1-0.2) node[left=2mm, pos=0.5]{\LARGE$m$};
 
    \draw [thick, decorate, decoration = {brace,amplitude=5pt,raise=0mm}] (-0.5,\k+\l+0.5) -- (-0.5, \k+1-0.5) node[left=2mm, pos=0.5]{\LARGE$\delta-1$};
	
\end{tikzpicture}

}
\caption{An example instance with uniform jobs at the moment of arrival of the third-to-last job (dashed). Yellow circles indicate the current assignment, white circles indicate the final assignment.}\label{fig:edf-uniform}
\end{figure}
\end{proof}
%

Finally, we consider a version of \edf with a limited number of allowed reassignments per iteration, analogously to $k$-\FF: $\kappa$-\edf works identically to \edf if an arriving job causes at most $\kappa=\kappa(n)$ reassignments, and rejects the job otherwise.
In contrast to \FF, limiting the number of reassignments of \edf to asymptotically~less than~$n$ reduces its competitivity.
Consider instances as in \Cref{fig:FFpyramid}, where the $i$-th job in the order of arrival has interval $[1,\ n-i+1]$ for $i\in \oneto{n}$ and causes reassignments of all previously accepted jobs. 
Here, $\kappa$-\edf accepts only~$\kappa+1$ first jobs and rejects the rest. 
Hence, for any $\kappa\in o(n)$, $\kappa$-\edf has no constant competitivity ratio. 

\newcommand{\alg}{\mathcal{A}}
\section{Batching per time unit}\label{sec:batching}
In the problem definition in the introduction, we assume that \textpat{}s are processed strictly in the order of their arrival.
Now we define an extended version of \probname, in which we are allowed to buffer the \textpat{}s arriving during one time unit and to assign them all at once, thus obtaining potentially better solutions. 

\begin{definition}[\probname with Batching, \probname-B]
In \emph{\probname with Batching}, the assignment is recomputed once per time unit.
That is, all \textpat{}s with the same arrival time are buffered and assigned simultaneously. 
\end{definition}
We call algorithms for \probname-B \emph{batching algorithms}. They additionally feature a deterministic behavior on input consisting of batches, i.e., sets of jobs.
We say that a non-batching algorithm~$\alg$ and a batching algorithm $\alg'$ \emph{correspond to each other}, if their behavior on single jobs is identical.

\begin{lemma}\label{lem:batching}
   Let $\alg$ be an algorithm for \probname with a competitive ratio smaller than $c$ for $c\in\Q$. Let $\alg'$ be a corresponding batching algorithm. Then $\alg'$ is at most $\frac{c+1}{2}$-competitive.
\end{lemma}
\begin{proof}
    Consider a worst-case instance $\Inst$ with $n$ jobs on which the algorithm $\alg$ yields a matching of cardinality $c\cdot n$.
    We show that there exists an instance $\Inst'$ with less than $2n$ jobs such that algorithm $\alg'$ makes the same number of rejections and reassignments on instance $\Inst'$ as algorithm $\alg$ on instance~$\Inst$.
    To this end, we transform the instance $\Inst$ into $\Inst'$ as follows. 
    For each time unit $\tp$ with $k>1$ arrivals, we split slot $\tp$ to $k$ consequent slots $\tp^{(i)}$, $i\in \{1,\ldots,k\}$.
    The jobs arriving at time $\tp$ now receive pairwise different arrival times $\tp^{(i)}$ in the order of their arrival.
    Next, we add $k-1$ dummy jobs $\left(-1, \tp^{(2)}, \tp^{(k)}\right)$ to the instance.
    The total number of dummy jobs is less than $n$. 
    Due to their intervals, they block all the additional slots created in the instance transformation, so that starting from time $0$, only original slots $\tp^{(1)}$ for $\tp \in \N$ and original jobs remain in the instance. 
    As the original jobs have pairwise different arrival times, algorithm $\alg'$ processes them identically to the corresponding non-batching algorithm $\alg$.   
\end{proof}

Hence, batching algorithms make, amortized, at least the half of the rejections and reassignments per job arrival compared to their counterparts without batching. 
For the algorithms considered in this work, however, we show a stronger result: Batching does not improve their competitive ratio.
\begin{lemma}
    The competitive ratios of \edf and \FF  with respect to the matching size cannot be improved by batching.
\end{lemma}
\begin{proof}
    
\edf is already optimal.
As for \FF, we slightly modify the worst-case instance family described in \Cref{lem:twotypesFF-UB} (see Figure~\ref{fig:FF-LBreject}):
We make the arrival times of long jobs pairwise different, and force thereby \FF with batching to behave identically to \FF without batching.
Formally, for some $\delta \in \N$ the instance contains the following jobs, indicated as tuples $(\req, \earlT, \lateT)$:
\begin{align*}
    \{(-i,\, 1,\, 3\delta -2) \mid 1\leq i \leq \delta-1\} &\cup \{ (i-1,\, i,\, i+\delta-1) \mid 1\leq i \leq \delta \} \\
    &\cup \{(\delta-1,\, \delta,\, 2\delta -1) \mid 1\leq i \leq \delta-1\}.
\end{align*}
On this instance, \FF with batching behaves as \FF and rejects $\delta-1$ of $3\delta-2$ jobs.
\end{proof}

Moreover, Lemma~\ref{lem:batching} shows that batching has no effect on our asymptotic results on numbers of reassignments.
\begin{corollary}
    Batching does not improve the worst-case asymptotic number of reassignments of any algorithm for \probname.   
\end{corollary}

\newcommand{\sset}[1][\pat]{S_{#1}}
\section{Relaxing the interval constraint}\label{sec:bmt}
A practical application -- scheduling of medical treatments -- motivated the interval constraints in \probname.
Now we consider a generalization of \probname in which we relax the interval constraints. 
In Bipartite Matching over Time (\bmt), each job $\pat$ has an arrival time $\req$ and a finite set $\sset \subset \N_{> \req}$ of feasible slots lying after the arrival time of a \textpat.
Note that we still allow for reassignments, so our results are independent of the \textsc{ranking}-algorithm due to Karp, Vazirani and Vazirani~\cite{Karp1990OnlineBM}. 
At the same time, our reassignments are restricted by the over-time property, in contrast to unconstrained reassignments considered by Bernstein, Holm and Rotenberg~\cite{Bernstein2019}.

Both \FF and \edf are well-defined in the generalized \bmt setting. 
For \edf, we define the \emph{deadline} $\lateT$ of a \textpat \pat as the latest slot of the set $\sset$.
In the following, we discuss the competitivity of both algorithms for \bmt.

Let us start with \edf, which is optimal for the interval-constrained version of the problem. 
Without interval constraints, our definition of \edf as in Algorithm~\ref{alg:EDF} is not better than a greedy algorithm for \obm.
\begin{lemma}
    \edf for \bmt is at most $\frac{1}{2}$-competitive with respect to the matching cardinality.
\end{lemma}
\begin{proof}
Consider an arbitrary instance of \obm consisting of $n$ jobs arriving online and of $n$ slots; enumerate the slots from $1$ to $n$. 
We transform this instance into an instance of \bmt: we add a slot $n+1$ and add it to each job's set of slots.
Finally, we add an additional job $\pat_{n+1}$ with $\sset[\pat_{n+1}] = \{n+1\}$, which arrives before all other jobs.
As all jobs now have the same deadline, \edf sorts the jobs only in the order of arrival. 
That is, the algorithms makes no reassignments on this instance, since, upon each new job arrival, all jobs that are already present is the schedule have the same deadline and are not considered for the set $S$  of potential reassignments.   
Hence, \edf is equivalent to a greedy algorithm without reassignments and is at most $\frac{1}{2}$-competitive on this class of instances.
\end{proof}

Next, we consider the \FF algorithm.
The proof of $\frac{2}{3}$-competitivity in Theorem~\ref{lem:FF-LB} holds for \bmt as well, since the proof does not rely on the interval constraints. 

\begin{corollary}
\FF is $\frac{2}{3}$-competitive for \bmt with respect to the matching cardinality.
\end{corollary}

To show that this competitive ratio is tight, we prove an even stronger result.

\begin{theorem}\label{thm:cr-UB-bmt}
No deterministic algorithm for \bmt has a competitive ratio greater than $\frac{2}{3}$.
\end{theorem}
\begin{proof}
    Consider a deterministic algorithm on a following instance of \bmt with three jobs in total.
    First, jobs $1$ and $2$ with slot sets $\sset[1] = \{1,3\}$ and $\sset[2] = \{1,2\}$ arrive at time $\req[1] = \req[2] =0$ in that order; see \Cref{fig:bmt-worstcase}.

    \begin{figure}[bt]
        \centering
        \resizebox{0.75\textwidth}{!}{

\begin{tikzpicture}[yscale=-0.95, xscale=1.1]
	\def\k{3}
	\def\intclip{0.15}
    \def\yOfAxis{4}
	
	\tikzstyle{patint}= [
	line width=1.5pt, 
	|-|,
	draw]
	
	\tikzstyle{arrivenode}= [
		circle, 
		fill,
		minimum size=0.4mm,
		inner sep=0.8mm]
	\tikzstyle{assignmark}[black]= [
		fill=yellow,
		circle,
		thick,
		inner sep=1mm,
		draw=#1]
	
	\def\tmax{3}
    \draw[->, thick] (-0.98,\yOfAxis) -- (\tmax+0.7,\yOfAxis) node[right]{{\LARGE$t$}};
	
	\foreach \x in {0,...,\tmax}{
		\draw[loosely dotted, gray] (\x, 0.5) -- ++ ($(0,\yOfAxis-0.5)$);
	}

	\fill[gray, opacity=0.2] (-0.98, 0.4) rectangle (1, \yOfAxis);

    \draw[patint]  ($(0, 1) + (\intclip,0)$) -- ($(1, 1) + (-\intclip, 0)$);
    \draw[patint]  ($(2, 1) + (\intclip,0)$) -- ($(3, 1) + (-\intclip, 0)$);
    \node[arrivenode] at (-0.5, 1) {};
	\node[assignmark] at ($(1, 1) + (-0.5,0)$) {};

    \draw[patint]  ($(0, 2) + (\intclip,0)$) -- ($(2, 2) + (-\intclip, 0)$);
    \node[arrivenode] at (-0.5, 2) {};
	\node[assignmark] at ($(2, 2) + (-0.5,0)$) {};

    \draw[patint, dashed]  ($(1, 3) + (\intclip,0)$) -- ($(2, 3) + (-\intclip, 0)$);
    \node[arrivenode] at (0.5, 3) {};

\begin{scope}[xshift=16cm]
    \draw[->, thick] (-0.98,\yOfAxis) -- (\tmax+0.7,\yOfAxis) node[right]{{\LARGE$t$}};
	
	\foreach \x in {0,...,\tmax}{
		\draw[loosely dotted, gray] (\x, 0.5) -- ++ ($(0,\yOfAxis-0.5)$);
	}

	\fill[gray, opacity=0.2] (-0.98, 0.4) rectangle (1, \yOfAxis);

    \draw[patint]  ($(0, 1) + (\intclip,0)$) -- ($(1, 1) + (-\intclip, 0)$);
    \draw[patint]  ($(2, 1) + (\intclip,0)$) -- ($(3, 1) + (-\intclip, 0)$);
    \node[arrivenode] at (-0.5, 1) {};
	\node[assignmark] at ($(3, 1) + (-0.5,0)$) {};

    \draw[patint]  ($(0, 2) + (\intclip,0)$) -- ($(2, 2) + (-\intclip, 0)$);
    \node[arrivenode] at (-0.5, 2) {};
	\node[assignmark] at ($(1, 2) + (-0.5,0)$) {};

    \draw[patint, dashed]  ($(2, 3) + (\intclip,0)$) -- ($(3, 3) + (-\intclip, 0)$);
    \node[arrivenode] at (0.5, 3) {};
\end{scope}

\begin{scope}[xshift=8cm]
    \draw[->, thick] (-0.98,\yOfAxis) -- (\tmax+0.7,\yOfAxis) node[right]{{\LARGE$t$}};
	
	\foreach \x in {0,...,\tmax}{
		\draw[loosely dotted, gray] (\x, 0.5) -- ++ ($(0,\yOfAxis-0.5)$);
	}

	\fill[gray, opacity=0.2] (-0.98, 0.4) rectangle (1, \yOfAxis);

    \draw[patint]  ($(0, 1) + (\intclip,0)$) -- ($(1, 1) + (-\intclip, 0)$);
    \draw[patint]  ($(2, 1) + (\intclip,0)$) -- ($(3, 1) + (-\intclip, 0)$);
    \node[arrivenode] at (-0.5, 1) {};
	\node[assignmark] at ($(3, 1) + (-0.5,0)$) {};

    \draw[patint]  ($(0, 2) + (\intclip,0)$) -- ($(2, 2) + (-\intclip, 0)$);
    \node[arrivenode] at (-0.5, 2) {};
	\node[assignmark] at ($(2, 2) + (-0.5,0)$) {};

    \draw[patint, dashed]  ($(1, 3) + (\intclip,0)$) -- ($(2, 3) + (-\intclip, 0)$);
    \node[arrivenode] at (0.5, 3) {};
\end{scope}

\end{tikzpicture}

        }
        \caption{Worst-case instance of \bmt: the figure shows all three possible assignments of the first two jobs and the corresponding choice of the third, dashed job. The third job will be rejected.}
        \label{fig:bmt-worstcase}
    \end{figure}
    
    Suppose that both jobs are assigned by the algorithm; if not, a competitive ratio greater than $\frac{2}{3}$ cannot be obtained.
    The third job arrives at time $\req[3]=1$.
    With its arrival, the assignment of the first two jobs becomes fixed.
    Job $3$ has a singe feasible slot that is chosen by an adversary depending on the assignment of the first two jobs:
    if slot $2$ is occupied, then $\sset[3]=\{2\}$; otherwise, slot $3$ must be assigned to job $1$, so we set $\sset[3] = \{3\}$.
    In both cases, the third job will be rejected.

    By repeating the described block of three jobs multiple times one obtains instances of any desired size for which one third of the jobs will be rejected.
\end{proof}
Hence, \FF has the best possible competitive ratio for \bmt.

In \Cref{sec:ff}, we saw that if all jobs in a \probname instance have intervals of the same length, then \FF is optimal, i.e.~$1$-competitive.
For the non-interval case, 
we show that no deterministic algorithm can obtain a competitive ratio better than $\frac{2}{3}$, which makes \FF again an optimal algorithm.

\begin{theorem}\label{thm:btm_uniform}
    No deterministic algorithm for \bmt has a competitive ratio greater than $\frac{2}{3}$ even on instances with feasible sets of constant cardinality.
\end{theorem}
\begin{proof}
    Similarly to the proof of \Cref{thm:cr-UB-bmt}, we construct a family of instances such that for any assignment on the first part of the jobs, there exists an adversarial complement of the instance on which the algorithm rejects one third of all jobs. 
    The family of instances that we consider contains $6$ jobs, each with two feasible slots.
    The first four jobs are fixed in the construction, while the last two jobs are constructed depending on the assignment of the first four (or omitted). 

    The first four jobs $\mathcal{J} = \{1,2,3,4\}$ arrive at time zero and have feasible sets $S_1 = \{1,3\}$, $S_2 = \{1,4\}$, $S_3 = \{2,5\}$ and $S_4 = \{2,6\}$.
    If all four jobs are assigned, then, by construction, at least two of the slots in the interval $[3,6]$ will be assigned to jobs in $\{1,2,3,4\}$; let these slots be $s$ and $s'$.
    The last two jobs arrive at time point $2$, which fixes the assignment of jobs in $\mathcal{J}$.
    These jobs obtain feasible sets $S_5 = S_6 = \{s,s'\}$ and will therefore be rejected, as no reassignment of jobs in~$\mathcal{J}$ is possible. 
    Hence, any deterministic algorithm achieves a competitive ratio of at most $\frac{\abs{\mathcal{J}}}{6} = \frac{2}{3}$ in this case.

    If two or more jobs from $\mathcal{J}$ are rejected, then we add no further jobs to the instance, so the attained ratio to an optimal solution is at most $\frac{1}{2}$ in this case.
    
    If exactly one job is rejected, then at least one slot $s\in [3,6]$ is assigned to a job in $\mathcal{J}$.
    Then we add two more jobs to the instance that arrive at time point~$2$, thus again fixing the assignment of jobs in $\mathcal{J}$. 
    The two jobs have feasible sets $S_5 = S_6 = \{s,s'\}$, where $s'$ is an arbitrary slot from $[3,6]$ with $s'\neq s$.
    Since slot~$s$ is irrevocably occupied, only one of the jobs $\{5,6\}$ can be assigned; hence, in this case again only four jobs will be assigned in total.
        
    The existence of an optimal matching of size $6$ for any choice of the set $\{s, s'\}$, or of a matching of size $4$ for the job set $\mathcal{J}$, can be easily verified by checking the Hall's condition for the underlying bipartite graph.
    Hence, any deterministic algorithm achieves a competitive ratio of at most $\frac{2}{3}$.
    \begin{figure}
        \centering
        \resizebox{0.3\textwidth}{!}{

\begin{tikzpicture}[yscale=-0.95, xscale=1.1]
	\def\k{3}
	\def\intclip{0.15}
    \def\yOfAxis{7}
	
	\tikzstyle{patint}= [
	line width=1.5pt, 
	|-|,
	draw]
	
	\tikzstyle{arrivenode}= [
		circle, 
		fill,
		minimum size=0.4mm,
		inner sep=0.8mm]
	\tikzstyle{assignmark}[black]= [
		fill=yellow,
		circle,
		thick,
		inner sep=1mm,
		draw=#1]
	
	\def\tmax{6}
    \draw[->, thick] (-0.98,\yOfAxis) -- (\tmax+0.7,\yOfAxis) node[right]{{\LARGE$t$}};
	
	\foreach \x in {0,...,\tmax}{
		\draw[loosely dotted, gray] (\x, 0.5) -- ++ ($(0,\yOfAxis-0.5)$);
	}
        \foreach \x in {0, 1, ..., 6}{
        \node[below] at (\x-0.5, \yOfAxis+0.1) {\LARGE$\x$};
    }

	\fill[gray, opacity=0.2] (-0.98, 0.4) rectangle (2, \yOfAxis);

    \draw[patint]  ($(0, 1) + (\intclip,0)$) -- ($(1, 1) + (-\intclip, 0)$);
    \draw[patint]  ($(2, 1) + (\intclip,0)$) -- ($(3, 1) + (-\intclip, 0)$);
    \node[arrivenode] at (-0.5, 1) {};
	\node[assignmark] at ($(1, 1) + (-0.5,0)$) {};

    \draw[patint]  ($(0, 2) + (\intclip,0)$) -- ($(1, 2) + (-\intclip, 0)$);
    \draw[patint]  ($(3, 2) + (\intclip,0)$) -- ($(4, 2) + (-\intclip, 0)$);
    \node[arrivenode] at (-0.5, 2) {};
	\node[assignmark] at ($(4, 2) + (-0.5,0)$) {};

    \def\y{3}
    \draw[patint]  ($(1, \y) + (\intclip,0)$) -- ($(2, \y) + (-\intclip, 0)$);
    \draw[patint]  ($(4, \y) + (\intclip,0)$) -- ($(5, \y) + (-\intclip, 0)$);
    \node[arrivenode] at (-0.5, \y) {};
	\node[assignmark] at ($(2, \y) + (-0.5,0)$) {};

    \def\y{4}
    \draw[patint]  ($(1, \y) + (\intclip,0)$) -- ($(2, \y) + (-\intclip, 0)$);
    \draw[patint]  ($(5, \y) + (\intclip,0)$) -- ($(6, \y) + (-\intclip, 0)$);
    \node[arrivenode] at (-0.5, \y) {};
	\node[assignmark] at ($(6, \y) + (-0.5,0)$) {};

    \def\y{5}
    \draw[patint, dashed]  ($(3, \y) + (\intclip,0)$) -- ($(4, \y) + (-\intclip, 0)$);
    \draw[patint, dashed]  ($(5, \y) + (\intclip,0)$) -- ($(6, \y) + (-\intclip, 0)$);
    \node[arrivenode] at (2-0.5, \y) {};
    \def\y{6}
    \draw[patint, dashed]  ($(3, \y) + (\intclip,0)$) -- ($(4, \y) + (-\intclip, 0)$);
    \draw[patint, dashed]  ($(5, \y) + (\intclip,0)$) -- ($(6, \y) + (-\intclip, 0)$);
    \node[arrivenode] at (2-0.5, \y) {};
 
\end{tikzpicture}

        }
        \caption{One example from the worst-case instance family for \bmt with uniform feasible sets. The dashed jobs are constructed according to the assignment of the first four jobs: their feasible slots are irrevocably occupied. The dashed jobs will be rejected by any algorithm.}
        \label{fig:bmt-uniform}
    \end{figure}
\end{proof}

\begin{corollary}
    \FF is an optimal algorithm for \bmt, also for feasible sets of fixed cardinality.
\end{corollary}

Finally, we discuss the number of reassignments.
The lower bounds on the worst-case number of reassignments for \probname clearly transfer to the more general \bmt. 
Hence, also on instances of \bmt \FF reassigns $\Omega(n\log n)$ and \edf $\Omega(n^2)$ jobs.
The upper bound of $\O(n\log^2 n)$ reassignments for the SAP strategy due to Bernstein, Holm and Rotenberg \cite{Bernstein2019} might, as mentioned in \Cref{sec:ff}, also apply to our over-time setting, i.e., to \FF for \bmt.
For $k$-\FF, the linear bound in the number of reassignments obviously holds in the non-interval setting. 
Just as for \probname, this restricted version still remains $\frac{2}{3}$-competitive, due to the arguments presented in \Cref{sec:k-ff}.
Hence, $k$-\FF is an optimal algorithm for \bmt that makes at most a linear number of reassignments.

\section*{Acknowledgments}
The authors thank anonymous reviewers for their constructive feedback, which helped greatly improve the presentation of this paper.
This work was partially supported by DFG (German Research Foundation) – RTG 2236/2.

\bibliography{lit}

\clearpage
\appendix

%

\section{Proof of \Cref{lem:LB-NlogN}}\label{apx:monsterProof}
\renewcommand{\thetheorem}{\ref{lem:LB-NlogN}}
\begin{lemma}
	In the worst case, \FF makes $\Omega(n \log n)$ reassignments on instances with $n$ \textpat{}s, even on instances of Static \probname.
\end{lemma}
\begin{proof}
    \newcommand{\lenk}[1][k]{L_{#1}}

    Consider an {upper-triangle} instance with $N = 2^n$ \textpat{}s arriving at time $0$.
    Jobs are enumerated in the order of arrival, \textpat $i\in \oneto{N}$ has time window $[1, N+1-i]$, see Fig.~\ref{fig:FFpyramid}.

	We consider only the time interval $[1,N]$, since all \textpat time windows are included in this interval.
    We partition the set of jobs into $n+1$ \emph{phases}: phase $k\leq n$ contains $\frac{N}{2^{k}} = 2^{n-k}$ \textpat{}s with indices
	$N-\frac{N}{2^{k-1}}+1,\ \ldots,\ N-\frac{N}{2^{k}}$,
    and phase $n+1$ contains one last \textpat.
	
	\begin{claim*} When processing jobs of phase $k$ for $k\leq n$, \FF makes $\frac{N}{2^k}(2^{k-1}-1)$ reassignments, and at the end of phase $k$, the assignment $\ass_k$ has the following structure.
	The jobs that have arrived so far are divided into $2^k -1$ blocks of $\frac{N}{2^{k}}$ subsequent \textpat{}s each;
	block $b\in \N$ contains \textpat{}s $\frac{N}{2^k}(b-1)+1,\ \ldots, \ \frac{N}{2^k}b$.
	The time horizon $[1,N]$ of the instance is also divided into $2^k$ intervals of length $\frac{N}{2^k}$.
	Let $\lenk \coloneqq \frac{N}{2^k}$ denote the block size at the end of phase $k$.
	The jobs of block $b \leq 2^k -1$ are assigned to the $(b+1)$-th last interval in the order of arrival, i.e.,
	\[
	\ass_k \colon \pat_{(b,i)}^k \mapsto \lenk(2^k - b-1) + i,
	\]
	where $\pat_{(b,i)}^k = \lenk(b-1) + i$ is the $i$-th \textpat of the $b$-th block in phase $k$, $1\leq j\leq 2^k-1$, $1\leq i\leq \lenk$.
	\end{claim*}
    
	\begin{claimproof}
		Proof by induction.\\
		$k=1$: Clearly, the first $\frac{N}{2}$ \textpat{}s are assigned by \FF to the slots $[1, \frac{N}{2}]$ in the order of their arrival, i.e.,
		\[\ass_1 \colon [1, \frac{N}{2}] \to [1, \frac{N}{2}],\quad i\mapsto i.
		\]
		This corresponds to one block of \textpat{}s, assigned to the second last (that is, first) interval of length $\frac{N}{2}$. No reassignments are necessary.
		\\
		$k-1 \to k$:
		Assume that	upon arrival of \textpat{}s of phase $k$, the slot assignment $\ass_{k-1}$ has the structure described above;
		in particular, \textpat{}s of block 
		\[B_b \coloneqq \left\{\lenk[k-1](b-1)+1,\ldots,\lenk[k-1](b-1)\right\}\]
		 are assigned to the interval 
		\[I_b \coloneqq \left[\lenk[k-1](2^{k-1}-b-1)+1,\ \lenk[k-1](2^{k-1} -b)\right],
		\]
		where $b \in \{1,\ldots, 2^{k-1}-1\}$.
		We denote by $I_0\coloneqq [\lenk[k-1](2^{k-1}-1), N]$ the rightmost interval, which contains free slots.
		Observe that the intervals are enumerated from right to left. 
		
		First, we show that for any \textpat of phase $k$, at least $2^{k-1}-1$ reassignments are necessary.
		  Jobs of phase $k$ 
		have deadlines $\frac{N}{2^{k-1}}$ or less.
		Hence, they require slots within the leftmost interval $I_{2^{k-1}-1}$, which is completely occupied by \textpat{}s of block $B_{2^{k-1}-1}$.
		
		For any $b\in \oneto{2^{k-1}-1}$, \textpat{}s of block $B_b$ have a deadline not greater than the deadline of \textpat $\pat_{(b,1)}$, which is 
		\[N+1-\lenk[k-1](b-1)-1 = \lenk[k-1](2^{k-1} - b +1),\] which is the right end point of the interval $I_{b -1}$.
		Hence, for any $b\in \oneto{2^{k-1}-1}$, \textpat{}s of block $B_b$ can be reassigned to interval $I_{b-1}$ the latest, replacing \textpat{}s of block $B_{b-1}$.
		Consequently, any reassignment path from the interval $I_{2^{k-1} -1}$, where the time windows of new \textpat{}s are located, to the empty slots in the interval $I_0$ will pass through all intermediate blocks and intervals.
		That is, at least $2^{k-1}-1$ reassignments per \textpat of phase $k$ are necessary under any reassignment strategy.
		
		Next we describe the reassignments done by \FF.
		We split the blocks of \textpat{}s and the time intervals in half;
		block $B_b$ is thus split in half-blocks $B_b^1$ and $B_b^2$ of size $\lenk = \frac{1}{2}\lenk[k-1]$, which are assigned to the interval 
		\[ I_{b}^1 \coloneqq \left[\lenk[k-1](2^{k-1}-b-1)+1,\ \lenk[k-1](2^{k-1} -b -1 + \frac{1}{2})\right]
		\]
		and
		\[ I_{b}^2 \coloneqq \left[\lenk[k](2^{k}-2b-1)+1,\ \lenk[k-1](2^{k-1} -b)\right],
		\]
		respectively.
		The jobs in phase $k$ build a block $B_{2^{k-1}} = B_{2^{k-1}}$ of size $\lenk$.
		
		Now consider a reassignment sequence which shifts the $i$-th \textpat of each block one block to the right, assigning the \textpat of block $B_1$ to the $i$-th empty slot in interval $I_0$:
		\[
			\varphi \colon \pat_{(b,i)} \mapsto \ass_{k-1}(\pat_{(b-1,i)}) \qquad \text{for }i\in\oneto{\lenk}\text{ and } b \in \oneto{2^{k-1}}.
		\]
		Observe that only the first $\lenk$ \textpat{}s of each block are involved in the reassignments, which corresponds to the half-blocks $B^1_b$, $b\in \oneto{2^{k-1}-1}$.
		
		Assignment $\varphi$ is feasible: consider the $i$-th \textpat $\pat \coloneqq \pat^{k-1}_{(b,i)}$ of $B_b$ for some $i\in \oneto{\lenk}$ and $b\in \oneto{2^{k-1}}$.
		The deadline of \textpat $\pat = \lenk[k-1](b-1)+i$ is
		\[\lateT = N+1- \frac{N}{2^{k-1}}(b-1) - i = \frac{N}{2^{k-1}}(2^{k-1}-b) + \left(\frac{N}{2^{k-1}}-i+1\right),\]
		and the new slot is 
		\[
			\varphi(\pat) = \ass_{k-1}(\pat_{(b-1,i)}) = \lenk[k-1](2^{k-1} -(b-1) -1) + i =  \frac{N}{2^{k-1}}(2^{k-1}-b) +i \leq \lateT,
		\]
		since $i \leq \frac{N}{2^k} = \frac{1}{2}\cdot\frac{N}{2^{k-1}}$.
		 	
		The reassignment path contains $2^{k-1}-1$ \textpat{}s being reassigned. 
		Moreover, since reassignments are executed one by one for \textpat{}s in order of their arrival, the earliest possible slot is selected at each stage in each reassignment sequence.
		Hence, the described reassignment $\varphi$ is exactly the \FF reassignment.
		
		The slot assignment $\ass_{k}$ at the end of phase $k$ thus has the following structure:
the set of \textpat{} s is divided into $2^{k} -1$ blocks $\bar{B}_\ell$, $\ell \in \oneto{2^k-1}$, where
		$\bar{B}_\ell = B^2_b$ if $\ell = 2b$ and $\bar{B}_\ell = B^1_b$ if $\ell = 2b-1$.
		Hence we have
		\[
			\ass_{k}\colon \oneto{N-\frac{N}{2^k}} \to \oneto{N-\frac{N}{2^k}}, \qquad \pat\mapsto \begin{cases}
				\varphi(\pat),&\text{if } \pat \in B^1_b\text{ for some }b,\\
				\ass_{k-1}(\pat),&\text{otherwise}.
			\end{cases}
		\]
		In other words, blocks $\bar{B}_{\ell}$ for even $\ell = 2b$ remain assigned to intervals
		\[\bar{I}_{\ell} = I_{b}^2 = \left[\lenk(2^k-2b-1)+1,\ \lenk(2^k-2b) \right] = \left[\lenk(2^k-\ell-1)+1,\ \lenk(2^k-\ell) \right].\]
		For odd values $\ell = 2b-1$, $b\leq 2^{k-1}$, blocks  $\bar{B}_{\ell} = B_b^1$ are assigned according to assignment $\varphi$ to the intervals
		\begin{align*}
		\bar{I}_{2b-1} = I_{b-1}^1 &= \left[
			\lenk(2^k - 2(b-1)-2)+1,\ \lenk(2^k-2(b-1)-1)
		\right] \\
  &= \left[\lenk(2^k - \ell -1)+1,\ \lenk(2^k-\ell) 	\right].
		\end{align*}
		Hence, the assignment $\ass_k$ at the end of phase $k$ has the desired structure.	

  The $2^{k-1}-1$ blocks of size $\lenk$ are reassigned, yielding $\frac{N}{2^k}(2^{k-1}-1)$ reassignments in phase $k$.
 \end{claimproof}
	Finally, we compute the total number of reassignments performed by \FF.
	As shown in the claim, 
	in phase $k\leq n$, $\frac{N}{2^k}(2^{k-1}-1)$ reassignments are made.
	In the last phase $n+1$, the last \textpat with index $N$ and time window $\{1\}$ arrives, and causes $N-1$ reassignments.
	The total number of reassignments is thus 
	\begin{align*}
		\sum_{k=1}^{n}\frac{N}{2^k}(2^{k-1}-1) +N-1 &= \frac{N}{2}\sum_{k=1}^{n}(1-\frac{1}{2^{k-1}})  + N-1\\
		&= \frac{N}{2}(n-1) - \frac{N}{2}\sum_{k=1}^{n}\frac{1}{2^{k-1}} + N-1\\
		&= \O(N\cdot n) = \O(N \log N).
	\end{align*}
\end{proof}

%
%
%

\newpage
\section{Limitations of $k$-\FF for small $k$}\label{apx:k-FF}
Figure~\ref{fig:k-ff_exmp} shows family of instances for each $k\in \N$, on which $(k+1)$-\FF yields a maximum matching, while $k$-\FF assigns only $k+1$ out of $k+2$ jobs.
Formally, the instance is constructed as follows: the first $k+1$ jobs have intervals $[i, i+1]$ for $i\in\oneto{k+1}$; the last job has interval $[1,1]$.
All jobs have the same arrival time $\req[] = 0$.

The only augmenting path for the last \textpat of the instance requires $k+1$ reassignments; hence, $k$-\FF rejects the last job.
Observe that the instance is an instance of Static \probname, as all jobs arrive at the same time.
We conclude that $k$-\FF is not optimal for Static \probname, in contrast to the unrestricted \FF.

\begin{figure}[htb]
	\centering
    \resizebox{!}{3.6cm}{

\begin{tikzpicture}[yscale=-0.85, scale=0.8]
	\def\k{5}
	\def\l{1}
	\def\intclip{0.15}
	\def\tmax{7}

	\tikzstyle{patint}= [
	line width=1.5pt, 
	|-|,
	draw]
	
	\tikzstyle{arrivenode}= [
	circle, 
	fill,
	minimum size=0.4mm,
	inner sep=0.8mm]
	\tikzstyle{assignmark}[black]= [
	fill=yellow,
	circle,
	thick,
	inner sep=0.9mm,
	draw=#1]

	\draw[->, thick] (1-0.58,\k+\l+1+2.3) -- (\tmax+0.7,\k+\l+1+2.3) node[right]{{\LARGE$t$}};
	
	\foreach \x in {1,...,\tmax}{
		\draw[loosely dotted, gray] (\x, 1.5) -- ++ ($(0,\l+\k+1.7)$);
	}
	

	\foreach \i in {1,...,\k}{
		\draw[patint]  ($(\i, \l+\i+0.5) + (\intclip,0)$) -- ($(\i+\l+1, \l+\i+0.5) + (-\intclip, 0)$);
		\node[assignmark] at ($(\l+\i, \l+\i+0.5) + (-0.5,0)$) {};
  
		\draw[->,>=stealth',thick, red!90!black] ($(\l+\i, \l+\i+0.5) + (-0.5,0) + (0, -0.3)$) to[out=-45,in=-135] ($(\l+\i, \l+\i+0.5) + (0.5,0) + (0, -0.3)$);
	}

	\draw[patint]  ($(1, \l+\k+2) + (\intclip,0)$) -- ($(\l+1, \l+\k+2) + (-\intclip, 0)$);
 
	\draw [thick, decorate, decoration = {brace,amplitude=5pt,raise=0mm}] (-0.5,\l+\k+1-0.2) -- (-0.5,\l+1+0.2) node[left=2mm, pos=0.5]{\LARGE$k+1$};
	
\end{tikzpicture}

    }
\caption{Example instance to compare $k$-\FF for different values of $k$. Jobs arrive in the top-to-bottom order.}\label{fig:k-ff_exmp}
\end{figure}

%

\end{document}